\documentclass{lmcs}


\usepackage[T1]{fontenc}
\usepackage{alphabeta}
\usepackage{cmll}

\usepackage{fontawesome5}
\DeclareFontFamily{U}{fontawesome1}{}
\DeclareFontShape{U}{fontawesome1}{m}{n}{<->FontAwesome--fontawesomeone}{}

\usepackage{amssymb}
\usepackage{array}
\usepackage{csquotes}
\usepackage{mathtools}
\usepackage{todonotes}
\usepackage{xspace}

\usepackage{ebproof}
\NewDocumentCommand \drule {m} {\ensuremath{\mathrm{(#1)}}}
\ebproofset{
	left label template		= {\small\drule{\inserttext}},
	right label template	= {\small\drule{\inserttext}},
	center					= false,
}
\usepackage{proof}

\usepackage{tikz}
	\usetikzlibrary{calc}
	\usetikzlibrary{positioning}
	\tikzstyle{m} = [fill=white, inner sep=1, scale=0.75]
	\tikzstyle{n} = [fill=blue!70!green!30, inner sep=1, rounded corners]
	\tikzstyle{b} = [thick]

\usepackage[
	style=alphabetic,
	maxbibnames=8,
]{biblatex}
\AtEndPreamble{
\usepackage[capitalise]{cleveref}
\crefname{cor}{Corollary}{Corollaries}
\crefname{defi}{Definition}{Definitions}
\crefname{prop}{Proposition}{Propositions}
\crefname{lem}{Lemma}{Lemmas}
\crefname{prob}{Problem}{Problems}
\crefname{thm}{Theorem}{Theorems}
\crefname{nota}{Notation}{Notations}
\crefname{obs}{Observation}{Observations}
\crefname{conj}{Conjecture}{Conjectures}
\crefname{exa}{Example}{Examples}
\crefname{rem}{Remark}{Remarks}
}

\newcommand{\bsub}{\begin{enumerate}[(i)]}
\newcommand{\esub}{\end{enumerate}}

\NewDocumentCommand \ie 			{}		{\emph{i.e.}\@\xspace}
\NewDocumentCommand \Ie 			{}		{\emph{I.e.}\@\xspace}
\NewDocumentCommand \eg 			{}		{\emph{e.g.}\@\xspace}

\NewDocumentCommand \lam		{}		
	{\texorpdfstring{\ensuremath{\lambda}}{lambda}}
\NewDocumentCommand \lamI		{}		
	{\texorpdfstring{\ensuremath{\lambda\comb{I}}}{lambda-I}}

\NewDocumentCommand \defemph 	{m} 	{\emph{#1}}

\NewDocumentCommand \eqdef		{}		{ \coloneqq }
\NewDocumentCommand \eqbnf		{}		{ \mathrel{\Coloneqq} }
\NewDocumentCommand \eqcobnf	{}		
	{ \mathrel{\Coloneqq}_{\textrm{co-ind.}} }
\NewDocumentCommand \bnf		{ommo}
	{\[	\IfValueT{#1}{#1 \quad \ni \quad}
		#2 \eqbnf #3
		\IfValueT{#4}{\quad (#4)}
	\]}
\NewDocumentCommand \cobnf		{ommo}
	{\[	\IfValueT{#1}{#1 \quad \ni \quad}
		#2 \eqcobnf #3
		\IfValueT{#4}{\quad (#4)}
	\]}

\NewDocumentCommand \set 		{smo}	
	{\IfBooleanTF{#1}
		{ \left\{ #2 \IfValueT{#3}{ \ \middle|\ #3 } \right\} }
		{ \{ #2 \IfValueT{#3}{ \ |\ #3 } \} }
	}

\NewDocumentCommand \subf		{}		{ \subseteq_{\mathrm{f}} }
\RenewDocumentCommand \emptyset {}		{ \varnothing }
\RenewDocumentCommand \setminus {}		{ - }
\NewDocumentCommand	\parts		{m}		{ \mathcal P(#1) }
\NewDocumentCommand	\partsfin	{m}		{ \mathcal P_{\mathrm{fin}}(#1) }
\NewDocumentCommand	\Mfin		{m}		{ \mathcal{M}_{\mathrm{fin}}(#1) }
\NewDocumentCommand \bigdunion	{}		{ \coprod }
\NewDocumentCommand \pto		{}		{ \rightharpoondown }

\NewDocumentCommand \Nat 		{}		{ \mathbb{N} }
\NewDocumentCommand \Var		{}		{ \mathcal{V} }

\NewDocumentCommand \red 			{oo}	
	{ \rightarrow \IfValueTF{#1}
		{ \IfValueTF{#2}{ _{#1}^{#2} }{ _{#1} } }
		{ \IfValueT{#2}{ ^{#2} } } 
	}
\NewDocumentCommand \redd 			{oo}	
	{ \mathrel{\rightarrow\mathrel{\mkern-14mu}\rightarrow} \IfValueTF{#1}
		{ \IfValueTF{#2}{ _{#1}^{#2} }{ _{#1} } }
		{ \IfValueT{#2}{ ^{#2} } } 
	}

\NewDocumentCommand \Lam		{}		{ \Lambda }
\NewDocumentCommand \LamI		{}		{ \Lambda_\comb{I} }

\NewDocumentCommand \subst 		{om}	
	{[ #2 / \IfValueTF{#1}{#1}{x}] }

\NewDocumentCommand \comb		{m}		{ \mathsf{#1} }
\NewDocumentCommand \Om			{}		{ \comb{\Omega} }
\NewDocumentCommand \Bible 		{o} 		
	{ \text{\scalebox{0.8}{\faBible}}_{\,\IfValueTF{#1}{#1}{l}} }

\NewDocumentCommand \bred 		{} 		{ \red[β] }
\NewDocumentCommand \breds 		{} 		{ \redd[β] }

\NewDocumentCommand \rterms		{ss}
	{\mathord{
		\Delta_{\comb{I}}
		^{\IfBooleanTF{#1}{ \IfBooleanTF{#2}{\!(!)}{!} }{}}
	}}

\NewDocumentCommand \frsums 	{oss}
	{ \IfBooleanTF{#2}
		{ \IfBooleanTF{#3}
			{ \Nat[ \rterms** \IfValueT{#1}{(#1)} ] }
			{ \Nat[ \rterms* \IfValueT{#1}{(#1)} ] } }
		{ \Nat[ \rterms \IfValueT{#1}{(#1)} ] }
	}

\NewDocumentCommand \esum		{m}		{ \rs{0}_{#1} }

\NewDocumentCommand \rt			{m}		{ #1 }
\NewDocumentCommand \rb			{m}		{ \bar{#1} }
\NewDocumentCommand \rs			{m}		{ \mathbf{#1} }
\NewDocumentCommand \rbs		{m}		{ \bar{\mathbf{#1}} }

\NewDocumentCommand \msubst 	{mo}	
	{ \left\lbrace #1 \middle/ \IfValueTF{#2}{#2}{x} \right\rbrace }

\NewDocumentCommand \rsubst 	{smo}	
	{ \left\langle \IfBooleanTF{#1}{ #2 }{ \rb{#2} } \middle/ 
	\IfValueTF{#3}{ #3 }{ x } \right\rangle }

\NewDocumentCommand \resource	{}		{ \mathsf{r} }
\NewDocumentCommand \rred 		{} 		{ \red[\resource] }
\NewDocumentCommand \rredr 		{} 		{ \red[\resource][?] }
\NewDocumentCommand \rreds 		{} 		{ \redd[\resource] }

\NewDocumentCommand \nf 		{}		{ \mathrm{nf} }
\NewDocumentCommand \size		{m}		{ \mathsf{size}(#1) }
\NewDocumentCommand \sumsize	{m}		{ \mathsf{ssize}(#1) }

\NewDocumentCommand \supp		{sm}
	{ | \IfBooleanTF{#1}{#2}{\rs{#2}} | }

\NewDocumentCommand \msetlt		{}		{ \prec_\Nat }
\NewDocumentCommand \msetgt		{}		{ \succ_\Nat }

\NewDocumentCommand \memory		{}		{ \mathsf{m} }

\NewDocumentCommand \fv 		{}		{ \mathrm{fv} }
\NewDocumentCommand \pfv 		{}		{ \mathrm{pfv} }

\NewDocumentCommand \BT 		{m}		{ \mathrm{BT}(#1) }
\NewDocumentCommand \OT			{om}	
	{ \mathrm{OT}\IfValueT{#1}{^{#1}}(#2) }
\NewDocumentCommand \MT			{om}
	{ \IfValueTF{#1} {\OT[#1]{#2}} {\OT{#2}} }

\NewDocumentCommand \Appset		{}		{ \mathcal{A}_\memory }
\NewDocumentCommand \Appof		{om}
	{ \mathrm{App}_{\IfValueTF{#1}{#1}{\memory}}(#2) }
\NewDocumentCommand \eqAppof	{om}
	{ \mathrm{App}_{\IfValueTF{#1}{#1}{\memory}}^=(#2) }
\NewDocumentCommand \dirapp		{om}	
	{ \omega_{\IfValueTF{#1}{#1}{\memory}}(#2) }
\NewDocumentCommand \MTle		{}		{ \sqsubseteq }

\NewDocumentCommand \bigMTlub	{}		{ \bigsqcup }

\NewDocumentCommand \MTrees		{}		{ \mathbf{O} }
\NewDocumentCommand \Appsets	{}		{ \mathbf{A} }
\NewDocumentCommand \mtapp		{m}		{ {\mathcal A}(#1) }
\NewDocumentCommand \appmt		{m}		{ {\mathcal T}(#1) }

\NewDocumentCommand \TE 		{ssm}	
	{ \mathcal{T}_\memory^{
		\IfBooleanT{#1}{\IfBooleanTF{#2}{!}{\bullet}}
	} (#3) }
\NewDocumentCommand \TMT		{ssm}
	{ \IfBooleanTF{#1}
		{ \IfBooleanTF{#2}{ \TE**{\MT{#3}} }{ \TE*{\MT{#3}} } }
		{ \TE{\MT{#3}} } 
	}
\NewDocumentCommand \NFT		{ssm}
	{ \nf( \IfBooleanTF{#1}
		{ \IfBooleanTF{#2}{ \TE**{#3} }{ \TE*{#3} } }
		{ \TE{#3} } 
	) }

\NewDocumentCommand \TEin		{}		{ \in^\bullet }
\NewDocumentCommand \TElt		{}		{ \subseteq^\bullet }
\NewDocumentCommand \TElub		{}		{ \cup^\bullet }
\NewDocumentCommand \bigTElub	{}			
	{ \mathop{\bigcup\nolimits^{\mathrlap{\bullet}}}\displaylimits }

\NewDocumentCommand \resourceof	{sm}	
	{ [\![#2]\!]_\resource^{\IfBooleanT{#1}{!}} }

\NewDocumentCommand \cB			{}		{ \mathcal{B} }
\NewDocumentCommand \cT			{}		{ \mathcal{T} }
\NewDocumentCommand \cO			{}		{ \mathcal{O} }
\NewDocumentCommand \cM			{}		{ \cO }
\NewDocumentCommand \cX			{}		{ \mathcal{X} }
\NewDocumentCommand \cY			{}		{ \mathcal{Y} }
\NewDocumentCommand \model		{m}		{ \mathcal{#1} } 

\NewDocumentCommand \Atoms		{}		{ \mathcal A }
\NewDocumentCommand \atom 		{}		{ \ast }

\NewDocumentCommand \lto		{}		{ \multimap }
\NewDocumentCommand \Types		{}		{ \mathbb{T} }
\NewDocumentCommand \OTypes		{}		{ \mathord{?\mathbb{T}} } 

\NewDocumentCommand \MTypes		{}		{ \mathord{!\Types} }
\NewDocumentCommand \vdashbang	{}		{ \vdash^! }

\DeclareMathOperator \domain			{dom}
\DeclareMathOperator \support			{supp}
\NewDocumentCommand \uimage		{}		{ \cup\mathrm{im} }
\NewDocumentCommand \ctxsep		{}		{ \mathrel{;} }
\NewDocumentCommand \derives	{}		{ \mathrel{\triangleright} }

\NewDocumentCommand \interp		{m}		{ [\![ #1 ]\!] }

\NewDocumentCommand \commentaire	{som}
	{\IfBooleanT{#1}{\vspace{\topsep}}%
	\colorbox {\IfValueTF{#2}{#2}{yellow}!50}
		{\IfBooleanTF{#1}
			{\parbox{.99\textwidth}{\sffamily#3}}
			{\sffamily #3}}%
	\IfBooleanT{#1}{\vspace{\topsep}}}
\definecolor{remy}{rgb}{1, 0.2, 0.4}

\NewDocumentEnvironment {oldstuff} {}
	{\color{gray}\colorlet{remy}{gray}}
	{}

\NewDocumentCommand \isep 			{}		{ \ |\ }

\NewDocumentCommand \Vars 			{}		{ \Var }

\NewDocumentCommand \LTERMS 		{oooo}	
	{ \IfValueT{#3}{#3} \Lambda \IfValueT{#4}{#4} \IfValueTF{#1}
		{ \IfValueTF{#2}{ _{#1}^{#2} }{ _{#1} } }
		{ \IfValueT{#2}{ ^{#2} } }
	}

\NewDocumentCommand \varstyle 		{m}		{ #1 }

\NewDocumentCommand \varx 			{} 		{ \varstyle{x} }

\NewDocumentCommand \vsetstyle 		{m}		{ #1 }
\NewDocumentCommand \vsetV 			{} 		{ \vsetstyle{V} }
\NewDocumentCommand \vsetW 			{} 		{ \vsetstyle{W} }
\NewDocumentCommand \vsetX 			{} 		{ \vsetstyle{X} }
\NewDocumentCommand \vsetY 			{} 		{ \vsetstyle{Y} }

\NewDocumentCommand \infrapp 		{mo}	
	{ \left(#1\right)^{\omega}\IfValueT{#2}{_{#2}} }

\NewDocumentCommand \lltstyle 		{m}		{ \mathsf{#1} }

\NewDocumentCommand \lltI 			{s} 	
	{ \lltstyle{I} \IfBooleanT{#1}{_{\emptyset}} }
\NewDocumentCommand \lltK 			{s} 	
	{ \lltstyle{K} \IfBooleanT{#1}{_{\emptyset}} }
\NewDocumentCommand \lltM 			{s} 	
	{ \lltstyle{M} \IfBooleanT{#1}{_{\vsetV}} }
\NewDocumentCommand \lltN 			{s} 	
	{ \lltstyle{N} \IfBooleanT{#1}{_{\vsetW}} }
\NewDocumentCommand \lltP 			{s} 	
	{ \lltstyle{P} \IfBooleanT{#1}{_{\vsetX}} }
\NewDocumentCommand \lltQ 			{s} 	
	{ \lltstyle{Q} \IfBooleanT{#1}{_{\vsetY}} }

\NewDocumentCommand \hreds 			{} 		{ \red[\mathrm{h}][*] }

\NewDocumentCommand \rtstyle 		{m}		{ #1 }
\NewDocumentCommand \rts 			{} 		{ \rtstyle{s} }
\NewDocumentCommand \rtt 			{} 		{ \rtstyle{t} }

\NewDocumentCommand \nlsubst 		{mo}	
	{ \left\lbrace #1 \middle/ \IfValueTF{#2}{#2}{\varx} \right\rbrace }

\NewDocumentCommand \lsubst 		{mo}	
	{\!\left\langle #1 \middle/ \IfValueTF{#2}{#2}{\varx} \right\rangle }

\NewDocumentCommand \mlsubst 		{mo}	
	{ \left\langle #1 \middle/ \IfValueTF{#2}{#2}{\varx} \right\rangle }

\usepackage{hyperref}

\begin{document}

\title[Ohana trees, linear approximation and multi-types]
	{Ohana trees, linear approximation and multi-types
	\texorpdfstring{\\}{}%
	for the \MakeLowercase{λ}I-calculus:
	\texorpdfstring{\\}{}%
	No variable gets left behind or forgotten!
	}
\thanks{This article is an extended version of the conference paper \autocite{Cerda.Man.Sau.25}.}

\author[R.~Cerda]{Rémy Cerda\lmcsorcid{0000-0003-0731-6211}}[a,b]
\thanks{The first and third authors were partly funded by
	the ANR project RECIPROG (ANR-21-CE48-019).
	The second author is partly funded by the EMERGENCE project Bang! of the Université Paris Cité.}

\author[G.~Manzonetto]
	{Giulio Manzonetto\lmcsorcid{0000-0003-1448-9014}}[a]

\author[A.~Saurin]{Alexis Saurin\lmcsorcid{0009-0002-1304-5518}}[a,c]

\address{Université Paris Cité, CNRS, IRIF, F-75013, Paris, France}
\email{Remy.Cerda@math.cnrs.fr, Giulio.Manzonetto@irif.fr, Alexis.Saurin@irif.fr}
\address{Università di Bologna, Italy}
\address{INRIA Picube, Paris, France}

\begin{abstract}
	Although the \lamI-calculus is a natural fragment of the \lam-calculus, 
	obtained by forbidding the erasure of arguments, 
	its equational theories did not receive much attention. 
	The reason is that all proper denotational models studied in the literature equate all non-normalizable \lamI-terms, whence the associated theory is not very informative. 
	The goal of this paper is to introduce a previously unknown theory of the \lamI-calculus, induced by a notion of evaluation trees that we call \enquote{Ohana trees}.
	The Ohana tree of a \lamI-term is an annotated version of its Böhm tree, remembering all free variables that are hidden within its meaningless subtrees, or pushed into infinity along its infinite branches.
	
	We develop the associated theories of program approximation: 
	the first approach --- more classic --- is based on finite trees and continuity, the second adapts Ehrhard and Regnier's Taylor expansion. 
	We then prove a Commutation Theorem stating that the normal form of the Taylor expansion of a \lamI-term coincides with the Taylor expansion of its Ohana tree. 
	As a corollary, we obtain that the equality induced by Ohana trees is compatible with abstraction and application. 

	Subsequently, we introduce a denotational model designed to capture the equality induced by Ohana trees. 
	Although presented as a non-idempotent type system, our model is based on a suitably modified version of the relational semantics of the \lam-calculus, which is known to yield proper models of the \lamI-calculus when restricted to non-empty finite multisets. 
	To track variables occurring in subterms that are hidden or pushed to infinity in the evaluation trees, we generalize the system in two ways: first, we reintroduce annotated versions of the empty multiset indexed by sets of variables; second, we introduce a separate environment to account for the free variables present in these subterms.
	We show that this model retains the standard quantitative properties of relational semantics and that the induced theory precisely captures Ohana trees.
	
	We conclude by discussing the cases of Lévy-Longo and Berarducci trees, 
	and possible generalizations to the full \lam-calculus.
\end{abstract}

\maketitle
\tableofcontents


\section{Introduction}

In the pioneering article
\enquote{The calculi of lambda-conversion}~\cite{Church41}
Alonzo Church introduced the \lam-calculus 
together with its fragment without weakening, 
called the \lamI-calculus, where each abstraction 
must bind \emph{at least} one occurrence of a variable. 
Historically, the \lamI-calculus has proven to be a useful framework for establishing results such as the finiteness of developments and standardization, which were successfully proven for the full \lam-calculus only decades later.
Despite that, in the last forty years there have been
no groundbreaking advances in its study.

The study of models and theories of \lam-calculus flourished in the 1970s~\cite{Bare} and remained central in theoretical computer science for decades~\cite{BarendregtM22}.
Although the equational theories of \lam-calculus, called \emph{\lam-theories}, form a complete lattice of cardinality $2^{\aleph_0}$~\cite{LusinS04}, only a handful are of interest to computer scientists.
Most articles focus on \lam-theories generated through some notion of \enquote{evaluation tree} of a \lam-term, which means that two \lam-terms  are equated in the theory exactly when their evaluation trees coincide. 
The definitions of evaluation trees present in the literature follow the same pattern: they collect in a possibly infinite tree all stable portions of the output coming out from the computation, and replace the subterms that are considered \emph{meaningless} by a constant $\bot$, representing the lack of information. 
By modifying the notions of \enquote{stable} and \enquote{meaningless}, 
one obtains Böhm trees~\cite{Barendregt77}, 
possibly endowed with some form of 
extensionality~\cite{Nakajima.75,SeveriV17,IntrigilaMP19}, 
Lévy-Longo trees~\cite{Levy76,Longo83} 
and Berarducci trees~\cite{Berarducci96}.
In the modern language of infinitary term rewriting systems, these trees can be seen as the transfinite normal forms of infinitary \lam-calculi~\cite{EndrullisHK12,Czajka20}.
Another popular method for defining \lam-theories is via observational equivalences~\cite{MorrisTh}: two \lam-terms are equivalent if they display the same behaviour whenever plugged in any context. 
For instance, the theory of Scott's $\mathcal{D}_\infty$ model~\cite{Scott72} captures the maximal consistent observational equivalence \cite{Hyland75,Wadsworth76}.

Since the \lam-calculus is an equationally conservative extension of the \lamI-calculus, every theory of the former is also a theory of the latter, namely a \emph{\lamI-theory}, while the converse does not hold. 
A typical example is the \lamI-theory $\mathcal{H}_{\comb{I}}$ generated by equating all \lamI-terms that are not $\beta$-normalizing. 
The maximal consistent observational equivalence $\mathcal{H}_{\comb{I}}^\eta$ is similar, except for the fact that it is extensional.
Note that equating all \lam-terms without a $\beta$-normal form gives rise 
to an inconsistent \lam-theory, but in the context of \lamI-calculus it is 
consistent and arises naturally: all \enquote{proper} denotational models of 
the \lamI-calculus studied so far induce either $\mathcal{H}_{\comb{I}}$ or 
$\mathcal{H}_{\comb{I}}^\eta$~\cite{HonsellL93}.
This is somewhat disappointing, since models and theories are most valuable when they capture non-trivial operational properties of programs, and normalization is relatively elementary in this setting.
In this paper we introduce a new \lamI-theory, induced 
by a notion of evaluation trees inspired by B\"ohm trees, but explicitly 
designed for \lamI-terms.

\paragraph*{B\"ohm trees and Taylor expansion for the \lamI-calculus}

An important feature of Böhm trees is that, because of their coinductive nature, they are capable of pushing some subterms into infinity. 
Consider for instance a \lamI-term $M_{xf}$ containing free variables $x,f$ and satisfying 
\[
	M_{xf} \redd[\beta] f(M_{xf})\redd[\beta] f(f(M_{xf}))\redd[\beta] f(f(f(M_{xf})))\redd[\beta] f^n(M_{xf})\redd[\beta]\cdots
\]
then the Böhm tree of $M_{xf}$ is the limit of this infinite 
sequence, namely $f^{\omega}=f(f(f(\cdots)))$.
Observe that $x$ is \enquote{persistent} in the sense that it is never 
erased along any finite reduction, but disappears in the limit. 
We say that $x$ is \emph{pushed into infinity} or \emph{forgotten}
in the Böhm tree of $M_{xf}$.
To some extent, also the variable $f$ is pushed into infinity, but it occurs infinitely often in~$f^{\omega}$. 
A variable $x$ can also disappear because it is hidden behind a meaningless term like $\Om = (\lam y.yy)(\lam y.yy)$. 
For instance, the B\"ohm tree of the \lamI-term $N = \Om x$ is just $\bot$ and we say that $x$ is \emph{left behind}.
As a consequence there are \lamI-terms, like $\lam xf.M_{xf}$ and $\lam xy.y(\Om x)$, whose Böhm tree is not a \lamI-tree~\cite[Def.~10.1.26]{Bare} since the variable $x$ is abstracted, but does not appear in the tree. 
This shows that Böhm trees are not well-suited for the \lamI-calculus, and the question of  whether other notions of trees can model \lamI-terms in a more faithful way naturally arises.

In the first part of this paper (\cref{sec:MemTrees})
we introduce a notion of Böhm trees keeping tracks 
of the variables persistently occurring along each possibly infinite path, 
or left behind a meaningless subterm.
We call them \emph{Ohana trees} as they remember all the variables present 
in the terms generating them, even if these variables are never actually 
used in their evaluations:
no variable is \emph{left behind} or \emph{forgotten}.
This is obtained by simply annotating such variables on the branches of the B\"ohm tree, and on its $\bot$-nodes, but bares interesting consequences: since Ohana trees are invariant under $\beta$-conversion, showing that two terms have distinct Ohana trees becomes a way of separating them modulo $=_\beta$.
Once the coinductive definition of Ohana trees is given 
(\cref{def:memorytrees4LamI}), we develop the corresponding theory 
of continuous program approximation based on finite trees. 
Perhaps surprisingly, in the finite approximants, the only variable annotations that are actually required are those on the constant $\bot$, as all other labels can be reconstructed when taking their  supremum.
Our main result in this setting is the Continuous Approximation Theorem
(\cref{thm:MTbijApp}) stating that the Ohana tree of a \lamI-term is 
uniquely determined by the (directed) set of all its finite approximants.

Inspired by quantitative semantics of linear logics, Ehrhard and Regnier 
proposed a theory of linear program approximation for regular \lam-terms 
based on Taylor expansion~ \cite{EhrhardR03,Ehrhard.Reg.08}.
A~\lam-term is approximated by a possibly infinite power-series of programs living in a completely linear \emph{resource calculus}, where no resource can be duplicated or erased along the computation.
This theory is linked to Böhm trees by a Commutation Theorem~\cite{Ehrhard.Reg.06} stating that the normal form of the Taylor expansion of a \lam-term is equal to the Taylor expansion of its Böhm tree.

The second part of our paper (\cref{sec:resource,sec:TE4MTs}) 
is devoted to explore the question of whether it is possible 
to define a Taylor expansion for the \lamI-calculus
capturing the Ohana tree equality.
In \cref{sec:resource} we design a \lamI-resource calculus with memory, 
mixing linear and non-linear features.
Our terms are either applied to non-empty \emph{bags} (finite multisets) of 
non-duplicable resources, or to an empty bag $1_X$ annotated with a set $X$ 
of variables (the non-linear parts of the terms). 
The reduction of a resource term $t$ preserves its free variables $\fv(t)$: if the reduction is valid then $t$ consumes all its linear resources and each variable is recorded in some labels, otherwise an exception is thrown and $t$ reduces to an empty program $\esum X$, where $X=\fv(t)$.

First we prove that the resource calculus so-obtained is confluent and strongly normalizing, then we use it as the target language of a Taylor expansion with memory, capable of approximating both \lamI-terms and Ohana trees. 
Our main result is that the Commutation Theorem from~\cite{Ehrhard.Reg.06} 
extends to this setting: the normal form of the Taylor expansion with memory 
of a \lamI-term always exists and is equal to the Taylor expansion with 
memory of its Ohana tree (\cref{the:commutation}).
As consequences (\cref{cor:MTeqIffNFTeq,cor:MisALamITheory}), 
we obtain that: 
\begin{enumerate}
\item The equality induced on \lamI-terms by Ohana trees coincides with the equality induced by the normalized Taylor expansion with memory; 
\item  The equality induced by Ohana trees is  compatible with application and  abstraction, in the sense of the \lamI-calculus, and it is therefore a \lamI-theory.
\end{enumerate}

\paragraph*{Denotational semantics of Ohana trees}
In the last part of this paper (\cref{sec:types}), 
we construct a denotational model of Ohana trees 
based on the relational semantics of linear logic. 
As is common, we present this model as 
a multi-type (or non-idempotent intersection type) system, 
where the interpretation of a \lamI-term is given 
by the set of its typing judgments. 
Finite multisets $[\alpha_1,\dots,\alpha_n]$ 
occur on the left-hand side of an arrow type:
$\Gamma\vdash M : [\alpha_1,\dots,\alpha_n] \lto \beta$ 
holds if $M$ needs to use its argument $n$ times 
to produce an output of type $\beta$. 
In our setting, we also have $\Gamma\vdash \lam x.M : []_{\vec y} \lto \beta$, expressing that $M$ may receive an argument whose free variables are $\vec y$, but will not use it to produce the output; instead, the argument is hidden or pushed into infinity.
A distinctive feature of our system is the introduction of a separate environment $\Delta$, as in $\Gamma;\Delta\vdash M : [\alpha_1,\dots,\alpha_n] \lto \beta$, which ensures proper $\alpha$-conversion by tracking the correspondence between the variables $x$ occurring in $M$ and the variables $\vec y$ of the \lamI-terms that are semantically substituted for $x$. 
We show that this type system satisfies quantitative versions of subject reduction and subject expansion, and that it validates our Taylor expansion (\cref{lem:typing:approx-TE-typings}). 
As a corollary of the Commutation Theorem, two \lamI-terms with the same Ohana tree share the same denotation, while the converse is established by inductive techniques. 
Hence, the model captures exactly the equality induced by Ohana trees (\cref{thm:interp-characterises-O}).

We consider this work a first step in a broader line of research, which we discuss further in the conclusion of the paper.
In particular, we aim to extend our formalism to the full \lam-calculus and explore its applicability to other kinds of trees.

\paragraph*{Related works}
On the syntactic side, the \lamI-calculus can be translated 
into MELL proof-nets without weakening 
studied in~\cite{CarvalhoF12,CarvalhoF16}. 
Our Ohana trees can be presented as labelled Böhm trees, and therefore share 
similarities with Melliès's (infinitary) \lam-terms \enquote{with 
boundaries}, whose branches are annotated by booleans~\cite{Mellies17}. 
However, our labels are describing additional points occurring at transfinite positions, a phenomenon reminiscent of the non-monotonic pre-fixed point construction defined in \cite{Berardid17a}. 
They can also be seen as an instance of the transfinite terms considered in \cite{KetemaetAl09}, but simpler, because no reduction is performed at transfinite positions.
Thus, Ohana trees should correspond to normal forms of an infinitary labelled \lamI-calculus, just like Böhm trees are the normal forms of the infinitary \lam-calculus~\cite{KennawayKSV95}.
We wonder whether a clever translation of the \lamI-calculus into the process calculus in \cite{DufourM24} would allow us to retrieve the Ohana trees and the associated Taylor expansion, in which case our Commutation Theorem might become an instance of their general construction.

Concerning denotational semantics, certain \emph{relevant} intersection type systems \cite{Bakel92,Ghilezan96,AlvesF22} can be used to represent models of the \lamI-calculus, as they describe reflexive objects in the SMCC of cpos and strict functions~\cite{Jacobs93}. In \cite{HonsellR92}, Honsell and Ronchi Della Rocca constructed a non-semi-sensible filter model that never equates a \lamI-term with a proper \lam-term; however, the induced notion of equivalence does not appear to have an obvious tree-like representation. In~\cite{HonsellL93}, the authors show that the model introduced in \cite{EgidiHR92} is fully abstract for the \lamI-theory $\mathcal{H}_{\comb{I}}^\eta$. The results in~\cite{CarvalhoF16} suggest that the relational semantics, endowed with the comonad of finite \emph{non-empty} multisets, contains models whose theory is either $\mathcal{H}_{\comb{I}}$ or $\mathcal{H}_{\comb{I}}^\eta$. No relational graph model can induce the Ohana tree equality, since they equate all \lam-terms having the same Böhm tree~\cite{PaoliniPR17,BreuvartMR18}. The relational model $\model{E}$ defined in \cite{CarraroES10}, using the comonad of finite multisets with possibly infinite coefficients, appeared more promising because it separates $M_{xf}$ from $M_{yf}$, and $\Om x$ from $\Om y$ when $x\neq y$; however, we will see that it still fails to capture the Ohana tree equality.


\section{Preliminaries}

This section is devoted to recall some basic notions and results concerning the \lam-calculus, and its fragment without weakening known as the \lamI-calculus. 
However, we assume that the reader is already familiar with the rewriting theory of \lam-calculus, its equational theories and its models based on B\"ohm trees.
For more details, we refer to Barendregt's book \cite{Bare}.

\subsection{The \lam-calculus, in a nutshell} 
Let us fix an infinite set $\Var$ of variables. The set $\Lam$ of 
\emph{\lam-terms} is defined by induction
\bnf[\Lam]{M,N}{ x \mid \lam x.M \mid MN }[\text{for $x \in \Var$}]
We use parentheses when needed.
We assume that application associates on the left, 
and has higher precedence than abstraction.
\emph{E.g.}, $\lam x.\lam y.\lam z.xyz$ stands for $\lam x.(\lam y.(\lam z.(xy)z))$.
We write $\lam x_1 \ldots x_n.M$ or even $\lam\vec x.M$ 
for $\lam x_1.\ldots \lam x_n.M$, 
and $N^n(M)$ for $N(N(\cdots N(M)\cdots))$, where $N$ occurs $n$ times. 
In particular, when $n = 0$, we get $\lam x_1\dots x_n.M = N^n(M) = M$.

The set $\fv(M)$ of \emph{free variables of $M$} is defined as usual, and we say that $M$ is \emph{closed} when $\fv(M) = \emptyset$.
Hereafter, \lam-terms are considered modulo \emph{$\alpha$-conversion}
(see \cite[\S 2.1]{Bare}).

\begin{exa}\label{ex:combinators}
The \lam-terms below are used throughout the paper to construct examples:
\[
	\comb{I} \eqdef \lam x.x,\quad 
	 \comb{K} \eqdef \lam xy.x,\quad 
	 \comb{F} \eqdef \lam xy.y,\quad 
	 \comb{B} \eqdef \lam fgx.f(g(x)),\quad
 	 \comb{D} \eqdef \lam x.xx,\quad
	 \Om \eqdef \comb{DD}.
\]
The \lam-term $\comb{I}$ is the identity, $\comb{K,F}$ are the projections, $\comb{B}$ is the composition, $\comb{D}$ is the diagonal and $\Om$ a looping term.
The \emph{pairing of $M,N$} is encoded as $[M,N] \eqdef \lam x.xMN$, for $x$ fresh. 
\end{exa}

\begin{defi}
\bsub
\item The \emph{$\beta$-reduction} $\to_\beta$ is generated by the rule
\[
	(\lam x.M)N\to M\subst{N}
\] where $(\lam x.M)N$ is called a \emph{$\beta$-redex} and $M\subst{N}$ stands for the \emph{capture-free substitution} of $N$ for all free occurrences of $x$ in $M$. 
\item We let $\redd[\beta]$ (resp. $=_\beta$) be the reflexive, transitive (and symmetric) closure of $\to_\beta$.
\item We say that $M$ \emph{is in $\beta$-normal form ($\beta$-nf)} if $M$ does not contain any $\beta$-redex.
\item We say that $M$ \emph{has a $\beta$-nf} if $M\redd[\beta] N$, for some $N$ in $\beta$-nf.
Since $\red[\beta]$ is confluent, such $N$ must be unique.
\esub
\end{defi}

\begin{rem}\label{rem:hnfs}
Any \lam-term $M$ can be written in one of the following forms: either $M = \lam\vec{x}.yM_1\cdots M_k$, in which case $M$ is called a \emph{head normal form} (hnf) and $y$ its \emph{head variable}; or $M = \lam\vec{x}.(\lam y.P)QM_1\cdots M_k$, where $(\lam y.P)Q$ is referred to as its \emph{head redex}.
\end{rem}
The \emph{head reduction} $\red[h]$ is the restriction of $\red[\beta]$ obtained by contracting the head redex.

\begin{defi} A \lam-term $Y$ is a \emph{fixed-point combinator} (fpc) if it satisfies $Y x =_\beta x(Yx)$.
\end{defi}
Notice that we do not require that fixed-point combinators are closed \lam-terms.
\begin{exa}\label{ex:fpcs} 
We recall Curry's fixed-point combinator $\comb{Y}$ and define, for all $l\in\Var$, 
a parametrized version of Turing's fpc $\Theta_l$ and Klop's \emph{biblic} fpc $\Bible$ introduced in \cite{ManzonettoPSS19}:
\begin{align*}
	\comb{Y} & \eqdef \lam f.\comb{D}_f\comb{D}_f,
		& \text{ where } \comb{D}_f \eqdef \lam x.f(xx), \\
	\Theta_l & \eqdef \comb{VV}l,
		& \text{ where } \comb{V} \eqdef \lam vlf.f(vvlf),\\	
	\Bible & \eqdef \lam e.\comb{BYB}el.
\end{align*}
It is easy to check that $\comb{Y}$ and $\Theta_l$ are fpcs.
Regarding the \emph{Bible} $\Bible$, which owes its name to its distinctive 
spelling, we get \[\Bible x =_\beta \comb{BYB}xl =_\beta 
\comb{Y}(\comb{B}x)l =_\beta \comb{B}x(\comb{Y}(\comb{B}x))l =_\beta 
x(\comb{Y}(\comb{B}x)l) =_\beta x(\Bible x).\]
Note that the variable $l$ remains in a passive position along the reduction of $\Bible$ (and of $\Theta_l$), therefore $\Bible[N]:= \Bible\subst[l]{N}$ (resp.\ $\Theta_N:= \Theta_l\subst[l]{N}$) remains a fpc, for all $N\in\Lambda$.
\end{exa}

Although fpcs $Y$ are not $\beta$-normalizable, they differ from $\Om$ since they produce increasing stable amounts of information along reduction: $Y\redd[\beta] \lam f.fY_1\redd[\beta] \lam f.f^n(Y_n)\redd[\beta]\cdots$.
Barendregt introduced a notion of \enquote{evaluation tree} to capture this infinitary behaviour~\cite{Barendregt77}.\par\addvspace{1.5ex}

\noindent\begin{minipage}{\textwidth}
\begin{defi}\label{def:BT} 
The \emph{Böhm tree} $\BT{M}$ of a \lam-term $M$ is coinductively defined as follows.
\begin{enumerate}
\item If $M$ has a hnf, that is $M \redd[h] \lam x_1\dots x_n.yM_1\cdots M_k$, then\\
\[
\begin{tikzpicture}
\node 		(BTM) at (-23pt,0)					{$\BT M \eqdef$};
\node[n] 	(y) 	at (1.1,0) 		{$\lam x_1\dots x_n.\,y$};
\node[n] 	(BTM1) 	at (.8,-.8) 	{$\BT{M_1}$};
\node 				at (2,-.8) 		{\dots};
\node[n] 	(BTM2) 	at (3.2,-.8) 	{$\BT{M_k}$};
\draw[b] ($(y.south east)+(-8pt,0)$) -- (BTM1.north);
\draw[b] (y.south east) -- (BTM2.north);
\end{tikzpicture}
\]
\item $\BT{M} \eqdef \bot$, otherwise. The constant $\bot$ represents the absolute lack of information.
\end{enumerate}
\end{defi}
\end{minipage}

\begin{exa}\label{ex:BTs} 
Using the \lam-terms from \cref{ex:combinators,ex:fpcs}, let us construct some Böhm trees.
\bsub
\item If $M$ has a $\beta$-normal form $N$ then, up to isomorphism, $\BT{M} = N$. \emph{E.g.}, $\BT{\comb{DI}} = \comb{I}$.
\item Since $\Om$ and $\comb{YK}$ have no hnf, we get $\BT{\Om} = \BT{\comb{YK}} = \bot$.
\item If $Y$ is an fpc, then $\BT{Y} = \lam f.f(f(f(\,\cdots)))$. Thus, $\BT{\comb{Y}} = \BT{\Bible} = \BT{\Theta_l}$.
\item\label{ex:BTs4} 
	Using the fpc $\comb{Y}$, one can define \lam-terms $\comb{E}_x,\comb{O}_x$ satisfying
	$\comb{E}_x =_\beta {[}\Om, [\Om x, \comb{E}_x]]$ and 
	$\comb{O}_x =_\beta {[}\Om x, [\Om, \comb{O}_x]]$. 
These two \lam-terms both construct a \enquote{stream} 
$[M_1,[M_2,[M_3,[\cdots]]]]$,
but the former places $\Om x$ at even positions and $\Om$ at odd ones, 
while the latter does the converse.
This difference is however forgotten in their Böhm trees:
\[
	\BT{\comb{E}_x} = \BT{\comb{O}_x} = [\bot,[\bot,[\bot,[\bot,[\,\cdots]]]]]\hfill\textrm{ (inline depiction of a tree)}
\]
\item Check that, given $\comb{R} = \lam ylx.x(y(xl))$, we have $\BT{\comb{YR}l} = \lam x_0.x_0(\lam x_1.x_1(\lam x_2.x_2(\,\cdots)))$.
\esub
\end{exa}

The equational theories of the \lam-calculus are called \lam-theories 
and are defined as follows.
\begin{defi} A \emph{\lam-theory} $\cT$ is given by an equivalence relation 
$=_\cT$ on \lam-terms containing $\beta$-conversion and \emph{compatible} 
with abstraction and application, \ie: 
\[
	\infer{\lam x.M =_\cT \lam x.M'}{M =_\cT M'}\qquad
	\infer{MN =_\cT M'N}{M =_\cT M'}\qquad
	\infer{NM =_\cT NM'}{M =_\cT M'}
\]
\end{defi}

The B\"ohm tree equality induces an equivalence on \lam-terms: $M =_\cB N$ if and only if $\BT{M} = \BT{N}$. 
As shown in \cite{Levy76,Wadsworth78}, $\cB$ contains $\beta$-conversion and is \emph{compatible} with abstraction and application, whence it is a \emph{\lam-theory}.

\subsection{The \lamI-calculus}\label{ssec:lamI}

The set $\LamI$ of \lamI-terms is the subset of $\Lam$ consisting of those \lam-terms 
in which every abstraction binds \emph{at least} one occurrence of the abstracted variable.
Note that  $\red[\beta]$ induces a reduction on $\LamI$ since $M\in\LamI$ and $M \red[\beta] N$ entail $N\in\LamI$. 
In other words, the property of being a \lamI-term is preserved by $\red[\beta]$.
Similarly, the \lamI-calculus inherits from the \lam-calculus 
all the notions previously defined like head reduction, Böhm trees, etc.
For convenience, we consider an alternative definition of $\LamI$ simultaneously defining $\lamI$-terms and their sets of free variables.

\begin{defi}\label{def:LamI}\mbox{}
\bsub
\item For all $X \subseteq \Var$, $\LamI(X)$ is defined as 
	the smallest subset of $\Lam$ such that:
\[
	\infer{x\in\LamI(\set{x})}{}\qquad
	\infer{\lam x.M\in\LamI(X\setminus\set{x})}{M\in\LamI(X)\quad x\in X}\qquad
	\infer{MN\in\LamI(X\cup Y)}{M\in\LamI(X) & N\in\LamI(Y)}
\]
The above definition differs from the usual definition of (regular) 
\lam-terms with their set of free variables only in the side condition of 
the second clause (namely, \enquote{$x\in X$}).
\item Finally, the set of all \lamI-terms is given by 
	the disjoint union
	$\LamI = \bigdunion_{X\subseteq\Var}\LamI(X)$.
\esub
\end{defi}

\begin{rem}\label{rem:aboutLambdaI}\mbox{}
\bsub
\item Note that $M \in \LamI(X)$ means that $M$ is a \lamI-term such that $\fv{(M)}=X$.
As a consequence, if $M\in\LamI(X)$ then the set $X$ must be finite.
\item\label{rem:aboutLambdaI2} 
	Each set $\LamI(X)$ is closed under $\beta$-conversion.
\item Our examples $\comb{I},\comb{B},\comb{D},\Om,\comb{Y},\Bible,\Theta_l,\comb{E}_x,\comb{O}_x,\comb{YR}l$ are all \lamI-terms.
\item On the contrary, $\comb{K},\comb{F},\comb{YK}\notin\LamI$.
\esub
\end{rem}

\begin{defi}
A \emph{\lamI-theory} $\cT$ is given by an equivalence $\mathord{=_\cT} \subseteq \LamI^2$ containing $\beta$-conversion and compatible with application and abstraction. 
In the context of \lamI-calculus, 
the compatibility with abstraction becomes: 
$M =_\cT M'$ and $x\in\fv(M)\cap\fv(M')$ 
imply $\lam x.M =_\cT \lam x.M'$.
\end{defi}
The only subtlety in the definition above stands in the side condition on the compatibility with abstraction.
However, this small difference has significant consequences: \eg, the theory axiomatized by equating all terms without a $\beta$-nf is consistent as a \lamI-theory, but inconsistent as a \lam-theory.
The following theorem collects the main properties of the \lamI-calculus.
\begin{thm}\label{thm:LamIMainProps}\mbox{}
\bsub
\item\label{thm:LamIMainProps1} 
	Every computable numerical function is \lamI-definable~\cite[Thm.~9.2.16]{Bare}.
\item\label{thm:LamIMainProps2} 
	Strong and weak $\beta$-normalization coincide for the \lamI-calculus~\cite[Thm.~11.3.4]{Bare}.
\item\label{thm:LamIMainProps3} 
	Every \lam-theory is also a \lamI-theory, while the converse does not hold. 
\esub
\end{thm}

By \cref{thm:LamIMainProps3}, Böhm trees induce a \lamI-theory, 
but we argue that this theory does not respect the spirit of $\LamI$. 
\emph{E.g.}, the variable $l$ is never erased along the reduction of, say, 
$\Bible$, but it disappears in its B\"ohm tree. 
One says that $l$ is \emph{pushed to infinity} in $\BT{\Bible}$.
This shows that the Böhm tree of the \lamI-term $\lam l.\Bible\in\LamI$
is not a $\lamI$-tree (using the terminology of \cite{Bare}):
the outer \lam-abstraction \enquote{$\lam l$} does not bind 
any free occurrence of $l$ in $\BT{\Bible} = \lam f.f(f(f(\,\cdots)))$.


\section{Ohana trees}\label{sec:MemTrees}

We have seen that Böhm trees are not well-suited for representing 
\lamI-terms faithfully, because variables that are never erased along the 
reduction can be forgotten (\ie, pushed into infinity).
We now introduce a notion of evaluation tree that records all variables that are pushed along each path, 
 ensuring that none are forgotten, whether they remain in passive position or not.

\subsection{The coinductive definition}

For every finite set $X\subseteq\Var$, written $X\subf \Var$, we introduce a constant $\bot_X$ representing any looping $M\in\LamI(X)$. 
Intuitively, the \emph{Ohana tree} of a \lamI-term $M$ is obtained from $\BT{M}$ by annotating each subtree (also 
$\bot$) with the free variables of the \lam-term that generated it.

\begin{defi}\label{def:memorytrees4LamI}
The \emph{Ohana tree} of a term $M\in\LamI(X)$ is coinductively defined as follows:
\bsub
\item If $M \redd[h] \lam x_1\dots x_n.yM_1\cdots M_k$ with $y\in 
X\cup\set{\vec x}$, $M_i\in\LamI(X_i)$ (for $1\le i \le k$), then\smallskip
\[
\begin{tikzpicture}
\node 		(MTM) 	at (-22pt,0)				{$\MT M \eqdef$};
\node[n] 	(y) 	at (1.15,0) 		{$\lam x_1\dots x_n.y$};
\node[n] 	(MTM1) 	at (.6,-1) 	{$\MT{M_1}$};
\node 				at (2,-1) 		{\dots};
\node[n] 	(MTM2) 	at (3.4,-1) 	{$\MT{M_k}$};
\draw[b] ($(y.south east)+(-8pt,0)$) to[pos=0.5]node[m]{$X_1$} (MTM1.north);
\draw[b] (y.south east) to[pos=0.5]node[m]{$X_k$} (MTM2.north);
\end{tikzpicture}
\]
\item $\MT{M} \eqdef \bot_X$, otherwise. Recall that $X = \fv(M)$, 
since $M\in\LamI(X)$.
\esub
\end{defi}


\begin{figure}
	\centering
	\scalebox{.85}{%
	\begin{tikzpicture}[baseline]
		\node[n] 	(E-lx1) 					{$λx_1.x_1$};
		\node[left=0pt of E-lx1.west] {$\MT{\comb E_x} =$};
		\node[n] 	(E-bot1) 	at (-1,-1) 		{$\bot_\emptyset$};
		\node[n] 	(E-lx2) 	at (1.5,-1) 		{$λx_2.x_2$};
		\node[n] 	(E-bot2) 	at (0.5,-2) 		{$\bot_{\set x}$};
		\node[n] 	(E-lx3) 	at (3,-2) 		{$λx_3.x_3$};
		\node[n] 	(E-bot3) 	at (2,-3) 		{$\bot_\emptyset$};
		\node 		(E-inf) 	at (4.5,-3) 		{};
		\draw[b] ($(E-lx1.south)+(3pt,0)$) to[pos=.5] node[m]{$\emptyset$} (E-bot1);
		\draw[b] ($(E-lx1.south east)-(2pt,0)$) to[pos=.5] node[m]{$\set x$} (E-lx2);
		\draw[b] ($(E-lx2.south)+(3pt,0)$) to[pos=.5] node[m]{$\set x$} (E-bot2);
		\draw[b] ($(E-lx2.south east)-(2pt,0)$) to[pos=.5] node[m]{$\set x$} (E-lx3);
		\draw[b] ($(E-lx3.south)+(3pt,0)$) to[pos=.5] node[m]{$\emptyset$} (E-bot3);
		\draw[b,dotted] ($(E-lx3.south east)-(2pt,0)$) to[pos=.5] node[m]{$\set x$} (E-inf);
	\end{tikzpicture}
	\hspace{1.5cm}
	\begin{tikzpicture}[baseline,overlay]
			\node[n] 	(O-lx1) 					{$λx_1.x_1$};
		\node[left=0pt of O-lx1.west] {$\MT{\comb O_x} =$};
		\node[n] 	(O-bot1) 	at (-1,-1) 		{$\bot_{\set x}$};
		\node[n] 	(O-lx2) 	at (1.5,-1) 		{$λx_2.x_2$};
		\node[n] 	(O-bot2) 	at (0.5,-2) 		{$\bot_\emptyset$};
		\node[n] 	(O-lx3) 	at (3,-2) 		{$λx_3.x_3$};
		\node[n] 	(O-bot3) 	at (2,-3) 		{$\bot_{\set x}$};
		\node 		(O-inf) 	at (4.5,-3) 		{};
		\draw[b] ($(O-lx1.south)+(3pt,0)$) to[pos=.5] node[m]{$\set x$} (O-bot1);
		\draw[b] ($(O-lx1.south east)-(2pt,0)$) to[pos=.5] node[m]{$\set x$} (O-lx2);
		\draw[b] ($(O-lx2.south)+(3pt,0)$) to[pos=.5] node[m]{$\emptyset$} (O-bot2);
		\draw[b] ($(O-lx2.south east)-(2pt,0)$) to[pos=.5] node[m]{$\set x$} (O-lx3);
		\draw[b] ($(O-lx3.south)+(3pt,0)$) to[pos=.5] node[m]{$\set x$} (O-bot3);
		\draw[b,dotted] ($(O-lx3.south east)-(2pt,0)$) to[pos=.5] node[m]{$\set x$} (O-inf);
	\end{tikzpicture} 
	\hspace{5cm}
	\begin{tikzpicture}[baseline]
		\node[n] 	(YR-l0) 					{$λx_0.x_0$};
		\node[left=0pt of YR-l0.west] {\llap{$\MT{\comb Y\comb Rl} =$}};
		\node[n]	(YR-l1)		at (8pt,-1.1)		{$λx_1.x_1$};
		\node[n]	(YR-l2)		at (16pt,-2.2)		{$λx_2.x_2$};
		\node		(YR-inf)	at (24pt,-3.3)		{};
		\draw[b] ($(YR-l1)+(0,26pt)$) to[pos=.5] node[m]{$\set{l,x_0}$} (YR-l1);
		\draw[b] ($(YR-l2)+(0,26pt)$) to[pos=.5] node[m]{$\set{l,x_0,x_1}$} (YR-l2);
		\draw[b,dotted] ($(YR-inf)+(0,26pt)$) to[pos=.5] 
			node[m]{$\set{l,x_0,x_1,x_2}$} (YR-inf);
	\end{tikzpicture}%
	}
	\caption{Examples of Ohana trees of \lamI-terms.}\label{fig:MTrees}
\end{figure}

We sometimes use an inline presentation of Ohana trees, 
by introducing application symbols annotated with the sets $X_i$, as in 
$\lam x_1\ldots x_n.y\cdot^{X_1}T_1\cdots^{X_k}T_k$.
\begin{exa}\label{ex:MTs}\mbox{}
\bsub
\item $\MT{\comb{I}} = \lam x.x$, $\MT{\comb{D}} = \lam x.x\cdot^{\set{x}}x$, $\MT{\comb{B}} = \lam fgx.f\cdot^{\set{g,x}}(g\cdot^{\set{x}}x)$.
\item $\MT{\Om} = \MT{\lam y.\Om y} = \bot_{\emptyset}$, $\MT{\Om x} = \bot_{\set x}$. More generally: $\MT{\Om \vec x} = \bot_{\set{x_1,\dots,x_n}}$. 
\item  $\MT{\comb{Y}} = \lam f.f\cdot^{\set{f}}(f\cdot^{\set{f}}(f\cdot^{\set{f}}(\,\cdots)))$. 
	In fact, it is easy to check that for all \emph{closed} fpcs $Y\in\LamI$, we have $\MT{Y} = \MT{\comb{Y}} \cong \BT{\comb{Y}}$.
\item In $\MT{\Bible}$ the variable $l$ which is pushed into infinity becomes visible: $\MT{\Bible} =\MT{\Theta_l} = \lam f.f\cdot^{\set{f,l}}(f\cdot^{\set{f,l}}(f\cdot^{\set{f,l}}(\,\cdots)))$. Thus $\MT{\Bible} \neq \MT{\Theta_x}$ when $x\neq l$.
\esub
\end{exa}

Other examples of Ohana trees are depicted in \cref{fig:MTrees}.
Note that the Ohana trees of $\comb{E}_x$ and $\comb{O}_x$ are now distinguished.
The interest of ${\comb{YR}l}$ is that it pushes into infinity an increasing amount of variables: 
\[\comb{YR}l =_\beta \lam x_0.x_0((\comb{YR})(x_0l)) =_\beta \lam x_0.x_0(\lam x_1.x_1((\comb{YR})(x_1(x_0l))))\, {=_\beta} \cdots\]
So---at the limit---infinitely many variables are pushed into infinity (but only $l$ is free).

\begin{rem} If $M\in\LamI$ has a $\beta$-nf $N$, then there is an isomorphism $\MT{M}\cong N$.
In fact, since $N$ is finite, it is possible to reconstruct all the labels in $\MT{M}$.  
\end{rem}

We are going to show that Ohana trees are invariant under $\beta$-reduction 
(\cref{prop:InvariantUnderBeta} plus \cref{thm:MTbijApp}), and that the 
equality induced on \lamI-terms is a \lamI-theory 
(\cref{cor:MisALamITheory}).

\subsection{The associated continuous approximation theory}

We introduce a theory of program approximation that captures Ohana trees, by adapting the well-established approach originally developed for Böhm trees~\cite{Levy76,Wadsworth78} (see also \cite[Ch.~14]{Bare}).
We start by defining the approximants with memory, corresponding to finite Ohana trees.

\begin{defi}\label{def:approximants}\mbox{}
\bsub
\item \label{def:approximants1}
	The set of \emph{approximants with memory} is 
	$\Appset \eqdef \bigdunion_{X\subf \Var} \Appset(X)$, 
	where $\Appset(X)$ is defined by (for $A_i\in\Appset(X_i)$, $X_i\subseteq X$, 
	$y\in X\cup\set{\vec x}$ and $\vec x\in\bigcup_{i=1}^k X_i\cup\set{y}$):
	\bnf[\Appset(X)]{A}{ \bot_X \mid \lam x_1\ldots x_n.yA_1\cdots A_k}
	
\item We define $\mathord{\MTle} \subseteq \Appset^2$ as 
the least partial order closed under the following rules:
\[
	\infer{\bot_X\MTle A}{A\in\Appset(X)}\qquad
	\infer{\lam x_1\dots x_n.y\, A_1\cdots A_k \MTle \lam x_1\dots x_n.y\, A'_1\cdots A'_k}{A_i\MTle A'_i&A_i,A'_i\in\Appset(X_i)&\textrm{for all }1\le i \le k&\vec x\in\set{y}\cup\big(\bigcup_{i=1}^k X_i\big)}
\]
\esub
\end{defi}

Each $(\Appset(X), \MTle)$ is a pointed poset (partially ordered set) with bottom element $\bot_X$, while the poset $(\Appset,\MTle)$ is not pointed because it has countably many minimal elements.
Note that we do not annotate the application symbols 
in $A\in\Appset$ with the sets of variables $X_i$, 
since the labels in its $\bot$-nodes carry enough information 
to reconstruct the finite Ohana tree associated with $A$.

\begin{lem}
\label{lem:bijAppset_finmemtrees}\label{lem:bijAppsetfinmemtrees}
There is a bijection between $\Appset$ and the set of finite Ohana trees of \lamI-terms.
\end{lem}

	\begin{proof}
	Given $A\in\Appset$, define a \lamI-term $M_A$ 
	by structural induction on $A$:
	\begin{itemize}
	\item if $A = \bot_{\set{x_1,\dots,x_n}}$ then 
		$M_A\eqdef \Om x_1\cdots x_n$.
	\item if $A = \lam\vec x.yA_1\cdots A_k$ then 
		$M_A \eqdef \lam\vec x.y(M_{A_1})\cdots (M_{A_k})$.
	\end{itemize}
	Define a map $\iota$ by $\iota(A) \eqdef \MT{M_A}$.
	Conversely, let $M\in\LamI$ be such that $T\eqdef\MT{M}$ is finite. 
	Since $T$ is finite, we can construct an approximant $\iota^-(T)$ 
	by structural induction on~$T$. There are two cases:
	\begin{itemize}
	\item $T = \bot_X$, with $X = \fv(M)$. 
		In this case, define $\iota^-(T) \eqdef \bot_X$.
	\item If $T = \lam\vec x.y\cdot^{X_1}T_1\cdots^{X_k} T_k$ then 
		$M\redd[h] \lam\vec x.yM_1\cdots M_k$ with $T_i = \MT{M_i}$ 
		and $X_i = \fv(M_i)$, for all $i\,(1\le i\le k)$. 
		In this case, define $\iota^-(T) = 
		\lam\vec x.y(\iota^-(T_1))\cdots (\iota^-(T_k))$.
	\end{itemize}
	A straightforward induction shows $\iota(\iota^-(T)) = T$ and 
	$\iota^-(\iota(A)) = A$.
	\end{proof}

We now give the recipe to compute the set of approximants with memory of a \lamI-term.

\begin{defi}\mbox{}
\bsub
\item The \emph{direct approximant} of $M\in\LamI(X)$ is given by:
\[
	\dirapp{M} \eqdef 
	\begin{cases}
	 \lam x_1\ldots x_n.y\,\dirapp{M_1}\cdots \dirapp{M_k},&\textrm{if }M = \lam x_1\ldots x_n.yM_1\cdots M_k,\\	
	\bot_X,&\textrm{otherwise.}
	\end{cases}
\]
\item The set $\Appof{M}$ of \emph{approximants with memory of $M\in\LamI$} is defined by:
\[
	\Appof{M} \eqdef \set{ A\in\Appset \mid \exists N\in\LamI\,.\, M\redd[\beta]N \textrm{ and }A\MTle \dirapp{N}}
\]
\esub
\end{defi}

\begin{rem} \label{rem:appof-M-fv-M}
	For all $M \in \LamI(X)$, we have $\dirapp M \in \Appset(X)$ and $\Appof M \subseteq \Appset(X)$.
\end{rem}

\begin{exa}\mbox{}
\bsub
\item $\Appof{\comb{D}} = \set{\bot_\emptyset, \lam x.x\bot_{\set x}, 
		\comb{D}}$, 
	$\Appof{\Om} = \set{\bot_{\emptyset}}$%
	.
\item $\Appof{\comb{Y}} = \set{\bot_\emptyset} \cup 
		\set{ \lam f.f^n(\bot_{\set f}) \mid n\in\Nat} $.
\item $\Appof{\Bible} = \set{\bot_{\set{l}}} \cup 
		\set{ \lam f.f^n(\bot_{\set{f,l}}) \mid n\in\Nat} $.
\item $\Appof{\comb{YR}l} = \set{\bot_{\set{l}},
	\lam x_0.x_0(\bot_{\set{x_0,l}}),
	\lam x_0.x_0(\lam x_1.x_1(\bot_{\set{x_0,x_1,l}})), \dots}$.
\esub
\end{exa}

\begin{lem}\label{lem:DirApp}\mbox{}
\bsub
\item\label{lem:DirApp1} 
	For every $M,N\in\LamI$, $M\bred N$ entails  $\dirapp{M} \MTle \dirapp{N}$.
\item\label{lem:DirApp2} 
	For all $M\in\LamI(X)$, $\Appof{M}$ is an ideal of the poset $(\Appset(X),\sqsubseteq)$. More precisely:
	\begin{enumerate}
	\item
		non-emptiness: $\bot_X\in \Appof{M}$;
	\item
		downward closure: $A\in \Appof{M}$ and $A'\MTle A$ imply $A'\in\Appof{M}$;
	\item
		directedness: $\forall A_1,A_2\in\Appof{M},\exists A_3\in\Appof{M}$ such that $A_1\MTle A_3 \sqsupseteq A_2$.
	\end{enumerate}
\esub
\end{lem}

	\begin{proof}
	\bsub
	\item By induction on a derivation of $M\bred N$. By \cref{rem:hnfs}, 
		either $M$ has a head redex or it has a head variable.
		If $M = \lam \vec x.(\lam y.M')M_0M_1\cdots M_k$ then 
		$\dirapp{M} = \bot_{\fv(M)}$. 
		By \cref{rem:aboutLambdaI}\ref{rem:aboutLambdaI2}, 
		$M\bred N$ entails $\fv(N) = \fv(M)$, 
		whence $\dirapp{N}\in\Appset(\fv(M))$.
		This case follows since $\bot_{\fv(M)}$ is the bottom of 
		$\Appset(\fv(M))$.
		
		If $M = \lam \vec x.yM_1\cdots M_k$ then $\dirapp{M} = \lam \vec 
		x.y\,\dirapp{M_1}\cdots \dirapp{M_k}$.
		The $\beta$-reduction must occur in some $M_i\bred N_i$, whence 
		$\dirapp{N} = \lam \vec x.y\,\dirapp{M_1}\cdots \dirapp{N_i}\cdots 
		\dirapp{M_k}$. We conclude since $\dirapp{M_i}\MTle\dirapp{N_i}$ 
		holds, from the induction hypothesis.
			
		\item To prove that $\Appof{M}$ is an ideal, we need to check three 
		conditions.
		\begin{enumerate}
		\item
			Since $\bot_{X}\sqsubseteq \dirapp{M}$ for all 
			$M\in\Appset(X)$, we get $\bot_{X} \in \Appof{M}$.
		\item
			By definition of $\Appof{M}$, this set is downward closed.
		\item
			Let us take $A_1,A_2\in\Appof{M}$ and prove that 
			$A_1\sqcup A_2\in\Appof{M}$.
			We proceed by structural induction on, say, $A_1$.
			
			Case $A_1 = \bot_X$. Then $A_2\in\Appset(X)$ and $A_1\sqcup A_2 
			= A_2\in\Appof{M}$.
			
			Case $A_1 = \lam x_1\ldots x_n.y\,A'_1\cdots A'_k$. If $A_2 = 
			\bot_X$ then proceed as before. Otherwise $A_1\in\Appof{M}$ 
			because there is a reduction $M\redd[\beta] \lam x_1\ldots 
			x_n.y\,M'_1\cdots M'_k=:M_1$ with $A'_i\in\Appof{M'_i}$.
			Similarly, $A_2\in\Appof{M}$ because there is a reduction 
			$M\redd[\beta] \lam x_1\ldots x_{n'}.y\,M''_1\cdots 
			M''_{k'}=:M_2$ with $A_2 = \lam x_1\ldots x_{n'}.y\,A''_1\cdots 
			A''_{k'} $ and $A''_j\in\Appof{M''_j}$.
			By confluence of $\bred$, $n = n'$ and $k = k'$, and every 
			$M'_i,M''_i$ have a common reduct $M'''_i$.
			By \cref{lem:DirApp}\ref{lem:DirApp1}, 
			$A'_i,A''_i\in\Appof{M'''_i}$, so by IH $A'_i\sqcup 
			A''_i\in\Appof{M'''_i}$.
			Conclude by taking $A_3 \eqdef \lam x_1\ldots x_n.y\,(A'_1\sqcup 
			A''_1)\cdots (A'_k\sqcup A''_k)$.\qedhere
		\end{enumerate}
	\esub
	\end{proof}

The next result shows that $\beta$-convertible terms share the same set of approximants.\par\addvspace{1.5ex}

\begin{prop}\label{prop:InvariantUnderBeta}
	Let $M,N\in\LamI(X)$. If $M\red[\beta] N$, then $\Appof{M} = \Appof{N}$.
\end{prop}

	\begin{proof}
	To prove $\Appof{M} = \Appof{N}$, we show the two inclusions.
	
	($\subseteq$)
	Take $A\in\Appof{M}$, then there exists a reduction $M\redd[\beta] M'$, 
	such that $A\MTle\dirapp{M'}$.
	By confluence there exists an $N' \in \LamI(X)$ with reductions
	$M' \breds N'$ and $N \breds N'$.
	By \cref{lem:DirApp}, \cref{lem:DirApp1} applied to the former reduction, $A \MTle \dirapp{M'} \MTle \dirapp{N'}$,
	hence $A \in \Appof N$.
	
	($\supseteq$) Straightforward, since every reduct of $N$ is also a 
	reduct of $M$.
	\end{proof}

In order to relate Ohana trees with approximants (\cref{thm:MTbijApp}), we need to extend the order on approximants to Ohana trees and define the set of finite subtrees of an Ohana tree:

\begin{defi} Let ${\MTrees}=\set{\MT{M}}[M\in\LamI]$.
\bsub
\item Given an Ohana tree $T \in\MTrees$, the set of its \emph{free 
variables} is defined by
	\[
	\fv(T) \eqdef \begin{cases}
	\big(\bigcup_{i=1}^k X_i \cup \{y\}\big) \setminus \{x_1,\dots,x_n\},&\textrm{ if }T=\lam x_1\ldots x_n.y\cdot^{X_1}T_1\cdots^{X_k}T_k,\\
	X,&\textrm{ if }T=\bot_X.
	\end{cases}
	\]
	Notice that if $T = \MT{M}$ then $\fv(T) = \fv(M)$.
\item
	Define $\mathord{\MTle}\subseteq {\MTrees}^2$ coinductively by the following rules:
	\[
	\infer{\bot_X\MTle T}{\fv(T)=X}\qquad
	\infer{\lam x_1\ldots x_n.y\cdot^{X_1}T_1\cdots^{X_k}T_k
		\MTle \lam x_1\ldots x_n.y\cdot^{X_1}T'_1\cdots^{X_k}T'_k}
		{T_i\MTle T'_i&\textrm{for all }1\le i \le k}
	\]
	(That is, $\MTle$ is the greatest relation $R \subseteq \MTrees^2$ such that every pair in $R$ is the conclusion of an instance of one of these rules whose premises themselves belong to $R$.)
	Straightforward coinduction shows that $\MTle$ defines a partial order on $\MTrees$.
\item
	Given an Ohana tree $T \in\MTrees$, define the set of its \emph{finite 
	subtrees} by setting:
	\[
		T^\star \eqdef \set{\OT N} [N \in \LamI,\ \OT N\MTle T \text{ and $\OT N$ is finite}].
	\]
\esub
\end{defi}

\newcommand\pos[1]{{\sf Pos}(#1)}

\begin{defi} Let $A\in\Appset$ be an approximant.
\bsub
\item The set $\pos{A}$ of \emph{positions of $A$} is
	the subset of $\Nat^*$ inductively defined by:
	\begin{itemize}
		\item If $A=\bot_X$, then $\pos{A}=\set\epsilon$
			(the empty sequence);
		\item If $A= \lam x_1\ldots x_n.yA_1\cdots A_k$, then $\pos{A}=\set{i\cdot p}[1\leq i \leq k \text{ and } p\in\pos{A_i}]$.
	\end{itemize}
\item Given a position $p\in \pos{A}$, we write $A(p)$ for the node of the syntactic tree of $A$ at position $p$. Formally, we define $A(p)\in\Appset$ by induction on the length of $p$ as follows:
	\begin{itemize}
		\item Base case $p=\epsilon$. If $A=\bot_X$ then ${A}(p)=\bot_X$,
		otherwise $A= \lam x_1\ldots x_n.yA_1\cdots A_k$ and  $A(p)=\lam x_1\ldots x_n.y$;
		\item Case $p=i\cdot p'$. Then $A= \lam x_1\ldots x_n.yA_1\cdots A_k$, 
		with $1\leq i \leq k$ and $A(p)=A_i(p')$.
	\end{itemize}
\esub
\end{defi}

Sets of positions are lifted to sets of approximants as expected and one can similarly define $\pos{M} = \pos{\Appof{M}}$. 
We can now prove the main theorem of this section, generalizing the 
Continuous Approximation Theorem for Böhm trees 
\autocite{Levy76,Hyland75,Wadsworth78}.

\begin{thm}[Continuous Approximation]\label{thm:MTbijApp} 
	There is a bijection $\mtapp{\cdot}$ between
	the set $\MTrees$ of Ohana trees 
	and the set $\Appsets = \set{\Appof{M}}[M\in\LamI]$
	of approximants of \lamI-terms,
	such that: 
	\[ 
		\mtapp{\MT{M}} = \Appof{M}, \textrm{ for all $M\in\LamI$}.
	\]
\end{thm}

\begin{proof}
	Let us build the requested bijection
	${\mathcal A}: \MTrees \rightarrow  \Appsets$.
	Let $M\in\LamI$ and $T$ its Ohana tree. 
	Since $T^\star$ contains only finite \enquote{subtrees} of $T$, 
	which are all Ohana trees of $\lamI$-terms, 
	one can consider a set of approximants $\mtapp T$ obtained 
	as the direct image of $T^\star$ by the mapping $\iota^-$
	defined in the proof of \cref{lem:bijAppsetfinmemtrees}. 
	In other words, we define $\mtapp{T} \eqdef \iota^-(T^\star)$. 
	We now prove that $\mtapp T = \Appof{M}$.
	
	To describe a finite portion of $T$,
	it is convenient to introduce \enquote{normal} contexts 
	with multiple holes denoted $[]_i$ (for $i\in\Nat$), 
	that are defined inductively by the following grammar:
	\[
		{\sf N} \eqbnf []_i \mid 
		\lam x_1\ldots x_n.y{\sf N}_1\cdots {\sf N}_k\qquad(i,n,k\in\Nat)
	\]
	Given a normal context $\sf N$ with $n$ holes 
	and terms $M_1,\ldots, M_n$, 
	we denote by ${\sf N}[M_1,\dots,M_n]$ 
	the term obtained from ${\sf N}$ by simultaneously substituting 
	each $M_i$ for  all the occurrences of the $i$-th hole, 
	possibly with capture of free variables.

	\begin{itemize}
	\item Let us prove that $\mtapp T \subseteq \Appof{M}$.
		Let $T'\in T^\star$, then there exist 
		$ M_1\in\LamI(X_1), \dots, M_n\in \LamI(X_n)$ and 
		a normal context ${\sf N}$ such that 
		${\sf N}[M_1,\dots, M_n] \in\LamI$ 
		with $M\redd[\beta] {\sf N}[M_1,\dots, M_n]$ 
		and $T = \MT{{\sf N}[\Om\vec{X_1},\dots, \Om\vec{X_n}]}$.
		Therefore $\iota^-(T') = {\sf N}[\bot_{X_1}, \dots, \bot_{X_n}]$ 
		and since $\bot_{X_i}\MTle \dirapp{M_i}$, $\iota^-(T') \in\Appof{M}$ 
		which proves the inclusion. 
	\item For the converse inclusion, 
		we prove by induction on the structure of $A$ that,
		for any $M\in\LamI$ and $A\in\Appof{M}$, 
		we have $A\in\mtapp{\MT M}$:
		\begin{itemize}
			\item If $A= \bot_X$, then $X = \fv(M)$ and 
			$A\in \iota^-(\MT{M}^\star)$;
			\item Otherwise, $A = \lam x_1\ldots x_n.yA_1\cdots A_k$. 
			In this case, we have 
			$M\redd[h] \lam x_1\ldots x_n.yM_1\cdots M_k$ 
			with $M_i\in\LamI(X_i)$ and $A_i\in\Appof{M_i}$ 
			for $1\leq i \leq k$. 
			By induction hypothesis we have some $T_i \in\MT{M_i}^\star$ 
			such that $A_i\in \iota^-(T_i)$ for all $i\,(1\leq i \leq k)$. 
			Since $\lam x_1\ldots x_n.y \cdot^{X_1} T_1 \cdots^{X_k} T_k 
			\in\MT{M}^\star$, we indeed have that $A \in \mtapp{\MT{M}}$.
		\end{itemize}
	\end{itemize}
	\medskip
	
	Let us now define ${\mathcal T}$, the inverse of ${\mathcal A}$.
	Let $A, B\in \Appof{M}$. 
	By directedness of $\Appof{M}$ 
	(\cref{lem:DirApp}\ref{lem:DirApp2}), 
	there exists $C\in\Appof{M}$ such  that $A,B\MTle C$.
	From this, it follows that at any position where 
	two approximants of $M$ are defined and different from $\bot_X$, 
	they agree. 
	We  define a (potentially) infinite tree $\appmt{\Appof{M}}$ 
	from $\Appof{M}$ as a pair of partial functions $(v,e)$, 
	respectively labelling the edges and vertices of the tree, 
	defined on $\pos{\Appof{M}}$ as follows.
	\begin{itemize}
		\item For every position $p \in \pos{\Appof{M}}$, 
		either for every $A\in \Appof{M}$ such that $p\in\pos{A}$, 
		$A(p)=\bot_X$, in which case $v(p) = \bot_X$, 
		or there is some $A\in \Appof{M}$ such that $p\in\pos{A}$, 
		$A(p)=\lam x_1\ldots x_n.y$ in which case 
		$v(p)\eqdef\lam x_1\ldots x_n.y$.
		\item For every non-empty position $p \in \pos{\Appof{M}}$,
		let $A\in \Appof{M}$ be such that $A(p)=\bot_X$ and 
		set $e(p)\eqdef X$.
	\end{itemize}
	A simple coinduction gives us 
	$\appmt{\mtapp{\MT{M}}}=\MT{M}$ for any $M\in\LamI$.
\end{proof}

\subsection{Digression: Ohana trees for the full \lam-calculus.}
\label{sec:Ohana-lamK}

It is natural to wonder whether it makes sense to 
extend the definition of Ohana trees to arbitrary \lam-terms.
The idea is to consider the set of 
\emph{permanent free variables} of $M\in\Lam$:
\[
	\pfv(M) \eqdef \set{ x \mid \forall N\in\Lam\,.\, M\redd[\beta] N 
	\textrm{ implies }x\in\fv(N)}
\]
Intuitively, $\pfv(M)$ is the set of free variables of $M$ that are never 
erased along any reduction.

\begin{defi}\label{def:memtrees4Lambda}
The \emph{Ohana tree} of a \lam-term $M$ can be then defined by setting:
\begin{itemize}
\item $\OT{M} \eqdef
	\lam\vec x.y\cdot^{\pfv(M_1)}\MT{M_1}\cdots^{\pfv(M_k)}\MT{M_k}$, 
	if $M \redd[h] \lam x_1\dots x_n.yM_1\cdots M_k$;
\item $\OT M \eqdef \bot_{\pfv(M)}$, otherwise.
\end{itemize}
\end{defi}

For $M\in\LamI$ the definition above is equivalent to 
\cref{def:memorytrees4LamI}, since $\pfv(M) = \fv(M)$.
For an arbitrary $M\in\Lam$, one needs an oracle to compute $\pfv(M)$, 
but an oracle is already needed to determine whether $M$ has an hnf. 

It is easy to check that the above notion of Ohana tree remains invariant by $\red[\beta]$. The key property which is exploited to prove this result is that $\pfv(\lam\vec x.yM_1\cdots M_k) = \set{y}\cup\pfv(M_1)\cup\cdots \cup \pfv(M_k)$. 
The main problem of \cref{def:memtrees4Lambda} is that the equality induced 
on \lam-terms is not a \lam-theory because it is not compatible with 
application.
The following counterexample exploits the fact that, in general, $\pfv(MN) \neq \pfv(M)\cup\pfv(N)$.
Define: 
\[
	\comb{W} = \lam xy.yxy,\qquad 
	M = \lam z.\comb{W}(zv_1v_2)\comb{W},\qquad 
	N = \lam z.\comb{W}(zv_2v_1)\comb{W},\qquad\textrm{ for }v_1,v_2\in\Var.
\]
We have $M\red[h] \lam z.(\lam y.y(zv_1v_2)y)\comb{W} \red[h] M$ and $N\red[h] \lam z.(\lam y.y(zv_2v_1)y)\comb{W} \red[h] N$, whence $\MT{M} = \MT{N} = \bot_{\set{v_1,v_2}}$. 
By applying these terms to $\comb{K}$, we obtain $M\comb{K} \redd[\beta] \comb{W}v_1\comb{W}$ and $N\comb{K} \redd[\beta] \comb{W}v_2\comb{W}$, therefore $\MT{M\comb{K}} = \bot_{\set{v_1}} \neq \bot_{\set{v_2}} = \MT{N\comb{K}}$. 
We conclude that the equality induced on $\Lam$ is not a \lam-theory, because it is not compatible with application.
This counterexample is particularly strong because it is independent from the fact that Ohana trees are generalizations of B\"ohm trees and not, say, Berarducci trees. 
Indeed, the $M,N$ constructed above have the same 
Böhm tree, Lévy-Longo tree and Berarducci tree.


\section{A resource calculus with memory}
\label{sec:resource}


We now introduce the \emph{\lamI-resource calculus} refining the (finite) resource calculus~\cite{EhrhardR03}, best known as the target language of Ehrhard and Regnier's Taylor expansion~\cite{Ehrhard.Reg.08}. 
So, the resource calculus is not meant to be a stand-alone language, but rather another theory of approximations for the \lam-calculus.
Before going further, we recall its main properties.


We consider here the promotion-free fragment of the resource calculus introduced in \cite{PaganiT09}.
Its syntax is similar to the \lam-calculus, except for the applications that are of shape $\rt{s}\rb t$, where $\rb t = [\rt{t}_1,\dots,\rt{t}_n]$ is a multiset of resources called a \emph{bag}.
The resources populating $\rb t$ are linear as they cannot be erased or copied by $\rt s$, they must be used \emph{exactly once} along the reduction. 
When contracting a term of the form $s = (\lam x.\rt{s'})[\rt{t}_1,\dots,\rt{t}_n]$ there are two possibilities.
\begin{enumerate}
\item If the number of occurrences of $x$ in $\rt{s'}$ is exactly $n$, then each occurrence is substituted by a different $\rt{t}_i$. 
Since the elements in the bag are unordered, there is no canonical bijection between the resources and the occurrences of $x$. 
The solution consists in collecting all possibilities in a formal sum of terms, the sum representing an inner-choice operator.
\item If there is a mismatch between the number of occurrences and the amount of resources, then $s$ reduces to the empty-sum, $\esum{}$. From a programming-language perspective, this can be thought of as a program terminating abruptly after throwing an uncaught exception.
\end{enumerate}
The first-class citizens of the resource calculus are therefore \emph{finite sums} of resource terms, that are needed to ensure the (strong) confluence of reductions. Another important property is strong normalization, that follows from the fact that no resource can be duplicated.

\subsection{\lamI-resource expressions and \lamI-resource sums}

Our version of the resource calculus is extended with labels representing the memory of free variables that were present in the \lamI-term they approximate.
Just like in \cref{def:approximants}\ref{def:approximants1} 
we endowed the constant $\bot$ with a finite set $X$ of variables, 
here we annotate:
\begin{itemize}
\item the empty bag $1$ of resource terms, since an empty bag of approximants of $M$ should remember the free variables of $M$;
\item the empty sum $\esum{}$ of resource terms. Indeed, if a resource term vanishes during reduction because of the mismatch described above, its free variables should be remembered.
\end{itemize}

\begin{defi} \label{def:rterms}\mbox{}
\bsub\item
	For all $X \subf \Var$, the sets $\rterms(X)$ and $\rterms*(X)$
	are defined by induction:
	\begin{gather*}
		\infer{ x \in \rterms(\set{x}) }{} \qquad
		\infer{ \lam x.\rt{s} \in \rterms(X \setminus \set{x}) }
			{ \rt{s} \in \rterms(X) & x \in X } \qquad
		\infer{ \rt{s}\rb{t} \in \rterms(X \cup Y) }
			{ \rt{s} \in \rterms(X) & \rb{t} \in \rterms*(Y) } \\
		\infer{ 1_X \in \rterms*(X) }{} \qquad
		\infer{ [\rt{t}_0,\dots,\rt{t}_n] \in \rterms*(X)}
			{ \rt{t}_0 \in \rterms(X) & \cdots & \rt{t}_n \in \rterms(X) }
	\end{gather*}
\item
	The set $\rterms$ of \emph{\lamI-resource terms} and
	the set $\rterms*$ of \emph{bags} are given by:
	\[	\rterms \eqdef \bigdunion_{X \subf \Var} \rterms(X)
		\quad \text{and}\quad
		\rterms* \eqdef \bigdunion_{X \subf \Var} \rterms*(X). \]
\esub
\end{defi}
As a matter of notation, we let $\rterms**$ denote either $\rterms$ or $\rterms*$, indistinctly but coherently.
We call \emph{resource expressions} generic elements $s,t\in\rterms**$. 
We denote the union of two bags $\rb{t}, \rb{u}\in\rterms*(X)$ multiplicatively by $\rb{t} \cdot \rb{u}$,
whose neutral element is the empty bag $1_X$.

\begin{rem}\label{rem:aboutresourceterms}\mbox{}
\bsub
\item
In every \lamI-resource term of the form $\lam x.s$, the variable $x$ must occur freely in $s$: it may appear in the undecorated underlying term, or in the \enquote{memory} $X$ decorating $1_X$.
Therefore, it makes sense to define $\fv(\rt s) \eqdef X$ whenever $\rt{s} \in \rterms**(X)$.
\item\label{rem:aboutresourceterms2}
Each set $\rterms*(X)$ is isomorphic to the monoid of
multisets of elements of $\rterms(X)$.
Notice that $\rterms*$ is \emph{not} the set of all bags
of elements of $\rterms$, just its subset of bags
whose elements \emph{have the same free variables}
(and so inductively in the subterms).
\esub
\end{rem}

\begin{exa}\label{ex:resourceterms}\mbox{}
\bsub
\item The identity $\comb{I}$ belongs to $\rterms$, whereas the projections do not: $\comb{K},\comb{F}\notin\rterms$.
\item Note that $\lam xy.x1_{\emptyset}\notin\rterms$ since $y\notin\fv(x1_{\emptyset})$, but $\lam xy.x1_{\set{y}}\in\rterms$ because $y\in\set{y}$.
\item The terms $\comb{D}_0 = \lam x.x1_{\set{x}}$ and $\comb{D}_{n+1} = 
\lam x.x[x,\ldots,x]$, where the bag contains $n+1$ occurrences of $x$, are 
\lamI-resource terms. By 
\cref{rem:aboutresourceterms}\ref{rem:aboutresourceterms2}, we obtain:
\item $[x,y] \notin \rterms*$, while $[x, x[x], x[\lam y.y,\lam y.y[y]], (\lam y.y)[x],  \lam y.y1_{\set x}] \in \rterms*(\set{x})\subseteq\rterms*$.
\esub
\end{exa}

We consider resource expressions up to $\alpha$-equivalence, under the proviso that abstractions bind linear occurrences of variables as well as occurrences in the memory of empty bags. For instance, $\lam xyz.x1_{\set{x}}1_{\set{x,y,z}}$ and  $\lam x'y'z.x'1_{\set{x'}}1_{\set{x',y',z}}$ are considered $\alpha$-equivalent. 

\begin{defi}\label{def:frsums}
	For all $X \subf \Var$, the set $\frsums[X]**$
	of \emph{sums of \lamI-resource terms} (\enquote{\mbox{resource sums}}, for short)
	is defined as the $\Nat$-semimodule
	of finitely supported formal sums of expressions in $\rterms**(X)$,
	with coefficients in $\Nat$.
	Explicitly, it can be presented as:
	\bnf[ \frsums[X]** ]{ \rs{s},\rs{t} }
		{ \esum X \mid \rt r\mid \rt{r} + \rs{s} }
		[ \text{for $\rt{r} \in \rterms**(X)$} ]
	quotiented by associativity and commutativity of $+$,
	as well as neutrality of $\esum X$.
\end{defi}

Note that $\rterms**(X)\subseteq\frsums[X]**$, that is, resource expressions are assimilated to the corresponding
one-element sum.
The constructors of the calculus are extended
to resource sums by (bi)linearity, \ie
for $\rt{s} \in \rterms(X)$, $\rs{s} \in \frsums[X]$, 
$\rb{t} \in \rterms*(Y)$, $\rbs{t} \in \frsums[Y]*$,
$\rt{u} \in \rterms(Y)$ and $\rs{u} \in \frsums[Y]$, we have:
\[
	\begin{array}{llll}
	(\rt{s}+\rs{s}) \rb{t} \eqdef \rt{s}\rb{t} + \rs{s}\rb{t},&
	\esum X \rb{t} \eqdef \esum {X \cup Y},&
	\lam x.(\rt{s}+\rs{s}) \eqdef \lam x.\rt{s} + \lam x.\rs{s},&
	\lam x.\esum X \eqdef \esum {X \setminus \set{x}},\\[1ex]
	\rt{s} (\rb{t}+\rbs{t}) \eqdef \rt{s}\rb{t} + \rt{s}\rbs{t},&
	\rt{s}\esum Y \eqdef \esum {X \cup Y},&
	[\rt{u} + \rs{u}] \cdot \rb{t} \eqdef [\rt{u}]\cdot\rb{t}
		+ [\rs{u}]\cdot\rb{t},&	
	[\esum Y] \cdot \rb{t} \eqdef \esum Y.
	\end{array}
\]
Therefore, if $\esum X$ occurs in $\rs s$
not as a summand but as a proper subterm,
then $\rs s = \esum{\fv(\rs s)}$.

\subsection{Memory substitution and resource substitution}

While the usual finite resource calculus is completely linear,
the variables we store in the index of $1_X$ and $\esum X$ are not:
the \enquote{memory} $X$ remembers
the variables present in $X$, not their amounts, 
this is the reason why it is modelled as a set, not as a multiset. 
This consideration leads us to define two kinds of substitutions.

The (non-linear) \emph{memory substitution} of a set $Y\subf\Var$ for a variable $x$ in a \lamI-resource term $s$ does not interact with the linear occurrences of $x$ in $s$ (\ie, they remain unchanged), it just replaces the \enquote{memory} of $x$ in the empty bags with the content of $Y$.

\begin{defi} \label{def:msubst}\mbox{}
\bsub
\item For all $X,Y \subf \Var$ and $x \in \Var$, define
	\[
	X \msubst Y \eqdef \begin{cases}
		X \setminus \{x\} \cup Y, 
			& \text{if $x \in X$,} \\
		X, & \text{otherwise.}\\
	\end{cases}
	\]
	\item Given $\rt{s} \in \rterms**(X)$, the \emph{memory substitution} of  $x$ by $Y\subf\Var$ in $\rt{s}$ is the resource term $\rt{s} \msubst Y$ defined as follows (for $x \neq y$ and, in the abstraction case, $y \notin Y$):
	\begin{align*}
		x \msubst Y & \eqdef x, &
		(\rt{s}\rb{t}) \msubst Y & \eqdef 
			(\rt{s}\msubst Y)(\rb{t}\msubst Y),  \\
		y  \msubst Y & \eqdef y,  &
		1_X \msubst Y & \eqdef 1_{X\msubst Y}, \\
		(\lam y.\rt{s}) \msubst Y & \eqdef \lam y.\rt{s} \msubst Y,	&
		([\rt s] \cdot \rb t) \msubst Y & \eqdef
			[\rt s \msubst Y] \cdot (\rb t \msubst Y).
	\end{align*}
\esub			
\end{defi}

We now define the \emph{resource substitution} of a bag $\rb u$ for $x$ in 
$\rt s$, whose effect is twofold:
(1)~it non-deterministically replaces each linear occurrence of $x$ with a 
resource from the bag (as usual), and
(2)~it applies the memory substitution of $x$ by $\fv(\rb u)$ to the 
resulting sum of terms.

\begin{defi} \label{def:rsubst}
	Given $\rt s \in \rterms**(X)$, $x \in \Var$ and $\rb u \in \rterms*(Y)$,
	the \emph{resource substitution} of $x$ by $\rb u$ in $\rt s$
	is the resource sum $\rt s \rsubst u \in \frsums[X \msubst Y]**$
	defined as follows:
	\begin{align*}
		x \rsubst u & \eqdef 
			\begin{cases}
			\rt u,& \text{if $\rb u = [\rt u]$},\\
			\esum Y, & \text{otherwise},
			\end{cases} &
		(\rt{s}\rb{t}) \rsubst u & \eqdef 
			\sum_{\rb u = \rb v \cdot \rb w}
			(\rt{s}\rsubst v)( \rb{t}\rsubst w),\\
		y  \rsubst u & \eqdef 
			\begin{cases}
			y, & \text{if $\rb u = 1_Y$,} \\
			\esum {\set y}, & \text{otherwise},
			\end{cases}&
		1_X \rsubst u & \eqdef 
			\begin{cases}
			1_{X \msubst Y}, & \text{if $\rb u = 1_Y$}, \\
			\esum {X \msubst Y}, & \text{otherwise,} \\
			\end{cases}\\
		(\lam y.\rt{s}) \rsubst u & \eqdef 
			\lam y.\rt{s} \rsubst u, &
		([\rt s] \cdot \rb t) \rsubst u & \eqdef
			\sum_{\rb u = \rb v \cdot \rb w}
			[\rt s \rsubst v] \cdot (\rb t \rsubst w).
	\end{align*}
	with $x \neq y$ and in the abstraction case $y \notin Y$.
	We extend it to sums in $\frsums[X]**$ by setting:
	\begin{align*}
		\esum X \rsubst u & \eqdef \esum {X\msubst Y} &
		(\rt{s} + \rs{s}) \rsubst u & \eqdef 
			\rt{s} \rsubst u + \rs{s} \rsubst u \\
		\rt s \rsubst* {\esum Y} & \eqdef \esum {X \msubst Y} &
		\rt s \rsubst* {(\rb u + \rbs u)} & \eqdef
			\rt s \rsubst u + \rt s \rsubst* {\rbs u}.
	\end{align*}
\end{defi}

It is easy to verify that the definition above does indeed define a resource 
sum in $\frsums[X \msubst Y]**$, and that it is stable under the quotients 
of \cref{def:frsums}.

Let us now state a lemma that is not strictly needed,
but helps understanding resource substitution
and corresponds to the way it is usually
\enquote{packaged} in the standard resource calculus.

\begin{defi}The \emph{linear degree} $\deg_x(\rt s)\in\Nat$ of $x \in \Var$
	in some $\rt s \in \rterms**$ is defined by:
	\begin{align*}
	\deg_x(x) & \coloneqq 1, &
	\deg_x(\lam y.\rt s) & \coloneqq \deg_x(\rt s)
		\text{, wlog. $x\neq y$,} &
	\deg_x(1_X) & \coloneqq 0, \\
	\deg_x(y) & \coloneqq 0 ,&
	\deg_x(\rt s \rb t) & \coloneqq \deg_x(\rt s) + \deg_x(\rb t), &
	\deg_x([\rt s] \cdot \rb t) & \coloneqq \deg_x(\rt s) + \deg_x(\rb t).
	\end{align*}

\end{defi}

\begin{lem}\label{lem:aboutressubst}
	Consider $\rt s \in \rterms**(X)$, $x \in \Var$ and
	$\rb u \in \rterms*(Y)$,
	and write $n \coloneqq \deg_x(\rt s)$. Then
	\[	\rt s \rsubst u = 
		\begin{cases}
		\sum\limits_{\sigma \in \mathfrak{S}(n)}
			\rt s\left[ \rt u_{\sigma(1)}/x^{(1)}, \dots, 
			\rt u_{\sigma(n)}/x^{(n)} \right] \msubst Y,
			& \text{if the cardinality of $\rb u$ is $n$,} \\
		\esum {X \msubst Y}, & \text{otherwise,}\\
		\end{cases} 
	\]
	where: $\mathfrak{S}(n)$ is the set of permutations of $\set{1,\dots,n}$;
	$\rt u_1,\dots,\rt u_n$ is any enumeration of the elements in $\rb u$;
	$x^{(1)},\dots,x^{(n)}$ enumerate the occurrences of $x$ in $\rt s$; and
	$\rt s\left[ \rt u_{\sigma(1)}/x^{(1)}, \dots, \rt u_{\sigma(n)}/x^{(n)} \right]$
	is the \lamI-resource term obtained by substituting each element $\rt u_{\sigma(i)}$ for the occurrence $x^{(i)}$.
\end{lem}
	
	\begin{proof}
	Straightforward induction on $\rt s$.
	\end{proof}

The next Substitution Lemma concerns the commutation of resource substitutions, and is the analogue of \cite[Lemma~2.1.16]{Bare}.
Notice that the assumptions on the substituted variables
are stronger than in the usual resource calculus
\cite[Lemma~2]{Ehrhard.Reg.08},
where only $x \notin Z$ is required:
here we also require that the substituted variables
actually occur in the term,
in accordance with the \enquote{\lamI} setting.

\begin{lem}[Substitution Lemma]
\label{lem:resource-substitution}
	Given $\rt s \in \rterms**(X)$,
	$\rb u \in \rterms*(Y)$,
	$\rb v \in \rterms*(Z)$,
	and
	$x \in X \setminus Z$, $y \in X \cup Y \setminus \set x$,
	\[	\rt s \rsubst u \rsubst v [y] =
		\sum_{\rb v = \rb v' \cdot \rb v''} \rt s \rsubst* {\rb v'} [y]
		\rsubst*{ \rb u \rsubst* {\rb v''}[y] }. \]
\end{lem}

Before proving the lemma, let us state the following intermediary result.

\begin{lem} \label{lem:rsubst-x-notin-fv}
	Given $\rt s \in \rterms**(X)$, $\rb u \in \rterms*(Y)$
	and $x \notin X$,
	\[	\rt s \rsubst u = \left\{ \begin{array}{ll}
		\rt s, & \text{if $\rb u = 1_Y$,} \\
		\esum X, & \text{otherwise.}
		\end{array} \right. \]
\end{lem}

	\begin{proof}
	Straightforward induction on $\rt s$.
	\end{proof}

	\begin{proof}[Proof of \cref{lem:resource-substitution}]
	By induction on $\rt s$.
	\begin{itemize}
	\item Case $\rt s = x$.
		(It is the only possible case of a variable
		because of the hypothesis $x \in X$.)
		We have
		\[	\sum_{\rb v = \rb v' \cdot \rb v''} x \rsubst* {\rb v'} [y]
				\rsubst*{ \rb u \rsubst* {\rb v''}[y] }
			= x \rsubst* {1_Z}[y] \rsubst*{\rb u \rsubst v[y]}
			= x \rsubst*{\rb u \rsubst v[y]}, \]
		hence if $\rb u = [\rt u]$, then
		both sides of the equality are equal to $\rt u \rsubst v[y]$,
		otherwise they are equal to $\esum{Y \msubst Z [y]}$.
	\item Case $\rt s = λy.\rt s'$.
		Immediate by the induction hypothesis on $s'$.
	\item Case $\rt s = \rt s' \rb s''$.
		We have $X = \fv(\rt s) = \fv(\rt s') \cup \fv(\rb s'')$.
		There are several possible cases.
		\begin{itemize}
		\item If $x \in \fv(\rt s')$ and $x \in \fv(\rb s'')$,
			then it is possible to apply the induction hypothesis
			in both subterms.
			By the definition of resource substitution,
			\begin{align*}
			&\rt s \rsubst u \rsubst v [y] = \\
			& = \adjustlimits \sum_{\rb u = \rb u_l \cdot \rb u_r}
				\sum_{\rb v = \rb v_l \cdot \rb v_r}
				\left(\rt s' \rsubst* {\rb u_l} \rsubst* {\rb v_l} [y]\right)
				\left(\rb s'' \rsubst* {\rb u_r} \rsubst* {\rb v_r} [y]\right)\\
			& = \adjustlimits \sum_{\rb u = \rb u_l \cdot \rb u_r}
				\sum_{\rb v = (\rb v'_l \cdot \rb v''_l) 
					\cdot (\rb v'_r \cdot \rb v''_r)}
				\left(\rt s' \rsubst* {\rb v'_l} [y]
					\rsubst*{ \rb u_l \rsubst* {\rb v''_l}[y]}\right)
				\left(\rb s'' \rsubst* {\rb v'_r} [y]
					\rsubst*{ \rb u_r \rsubst* {\rb v''_r}[y]}\right) \\
			\intertext{by the induction hypotheses 
			on $\rt s'$ and $\rb s''$,}
			& = \sum_{\rb v = \rb v' \cdot \rb v''} \rt s 
				\rsubst* {\rb v'} [y]
				\rsubst*{ \rb u \rsubst* {\rb v''}[y] }
			\end{align*}
			by permuting the sums and re-arranging the indices
			(using associativity and commutativity of the multiset union).
		\item If $x \in \fv(\rt s)$ and $x \notin \fv(\rb s'')$, then again
			\begin{align*}
			\rt s \rsubst u \rsubst v [y]
			& = \adjustlimits \sum_{\rb u = \rb u_l \cdot \rb u_r}
				\sum_{\rb v = \rb v_l \cdot \rb v_r}
				(\rt s' \rsubst* {\rb u_l} \rsubst* {\rb v_l} [y])
				(\rb s'' \rsubst* {\rb u_r} \rsubst* {\rb v_r} [y]) \\
			& = \sum_{\rb v = \rb v_l \cdot \rb v_r}
				(\rt s' \rsubst* {\rb u} \rsubst* {\rb v_l} [y])
				(\rb s'' \rsubst* {1_Y} \rsubst* {\rb v_r} [y]) \\
				& \qquad + \sum_{\substack{
					\rb u = \rb u_l \cdot \rb u_r \\ \rb u_r \neq 1_Y }}
				\sum_{\rb v = \rb v_l \cdot \rb v_r}
				(\rt s' \rsubst* {\rb u_l} \rsubst* {\rb v_l} [y])
				(\rb s'' \rsubst* {\rb u_r} \rsubst* {\rb v_r} [y]) \\
			& = \sum_{\rb v = \rb v_l \cdot \rb v_r}
				(\rt s' \rsubst* {\rb u} \rsubst* {\rb v_l} [y])
				(\rb s'' \rsubst* {\rb v_r} [y])
				+ \esum{X \msubst Y \msubst Z[y]} \\
			\shortintertext{by \cref{lem:rsubst-x-notin-fv},}
			& = \sum_{\rb v = (\rb v'_l \cdot \rb v''_l) \cdot \rb v_r} 
				\left(\rt s' \rsubst* {\rb v'_l} [y]
				\rsubst*{ \rb u \rsubst* {\rb v''_l}[y] }\right)
				\left(\rb s'' \rsubst* {\rb v_r} [y]\right) \\
			\shortintertext{by the induction hypothesis on $\rt s'$,}
			& = \sum_{\rb v = (\rb v'_l \cdot \rb v_r) \cdot \rb v''_l}
				\rt s \rsubst* {\rb v'_l \cdot \rb v_r} [y]
				\rsubst* { \rb u \rsubst* {\rb v''_l}[y] }
			\end{align*}
			by re-arranging the indices,
			and using again \cref{lem:rsubst-x-notin-fv}
			together with the hypotheses
			$x \notin Z$ and $x \notin \fv(\rb s'')$.
			This is exactly the desired result.
		\item The case $x \notin \fv(\rt s)$ and $x \in \fv(\rb s'')$
			is symmetric.
		\end{itemize}
	\item Case $\rt s = 1_X$.
		If $\rb u = 1_Y$ and $\rb v = 1_Z$,
		both sides of the equality are equal to
		$1_{X \msubst Y \msubst Z [y]}$.
		Otherwise, they are equal to
		$\esum{X \msubst Y \msubst Z [y]}$.
	\item Case $\rt s = \rt s' \cdot \rb s''$
		is similar to the case of application above.
	\qedhere
	\end{itemize}
	\end{proof}

\subsection{The operational semantics}

We now equip resource expressions and resource sums with a notion of reduction, denoted $\rred$. We first define reduction on a single \lamI-resource term and then extend it to resource sums in the usual way. Note that contracting a single redex in a term $t \in \rterms**(X)$ already produces a finite sum of terms in $\frsums[X]**$, due to the way substitution is defined (\cref{lem:aboutressubst}).

\begin{defi} \label{def:rred}\mbox{}
\bsub
\item
	For each $X \subf \Var$, define the \emph{resource reduction} as a relation between \lamI-resource terms and resource sums, \ie
	$\mathord{\rred} \subseteq \rterms**(X) \times \frsums[X]**$:\\[1ex]
	$\infer{ (\lam x.\rt s)\rb t \rred \rt s \rsubst t }{} \qquad
		\infer{ \lam x.\rt s \rred \lam x.\rs s' }{ \rt s \rred \rs s' } \qquad
		\infer{ \rt s \rb t \rred \rs s' \rb t }{ \rt s \rred \rs s' } \qquad
		\infer{ \rt s \rb t \rred \rt s \rbs t' }
			{ \rb t \rred \rbs t' } \qquad
		\infer{ [\rt s] \cdot \rb t \rred {[\rs s']} \cdot \rb t }
			{ \rt s \rred \rs s' } $\smallskip
\item\label{def:rred3} We extend the reduction relation $\rred$ to resource sums $\frsums[X]** \times \frsums[X]**$, and simultaneously introduce its reflexive closure $\rredr$, as follows: 
\[
	\infer{\rt s + \rs t \rred \rs s' + \rs t' }{
		\rt s \rred \rs s'
		&
		\rs t \rredr \rs t'
		}
	\qquad
	\infer{\rs t \rredr \rs t'}{\rs t \rred \rs t'}	
	\qquad
	\infer{\rs t \rredr \rs t}{}
\]
\item
	We denote by $\redd[\resource]$ the reflexive and transitive closure of 
	the relation $\rred$ given in \ref{def:rred3}. 
\esub
\end{defi}

Observe that $\rred$ is well-defined because $\fv((\lam x.\rt s)\rb t) = \fv(\rt s \rsubst t)$.
Indeed $\fv((\lam x.\rt s)\rb t) = 
\fv(\rt s) \setminus \set x \cup \fv(\rb t)$
with $x \in \fv(\rt s)$, whence $\fv(\rt s \rsubst t) = \fv(\rt s) \msubst {\fv(\rb t)}
= \fv(\rt s) \setminus \set x \cup \fv(\rb t)$.

\begin{exa} We use the resource \lamI-terms $\comb{I}$ and $\comb{D}_n$ from 
\cref{ex:resourceterms}.
\bsub
\item $\comb{D}_0[z]\rred (x1_{\set{x}})\langle[z]/x\rangle = 
	(x\langle[z]/x\rangle)(1_{\set{z}}) + (x\langle1_{\set{z}}/x\rangle)(1_{\set x}\langle [z]/x\rangle) = z1_{\set z}+\esum{\set{z}} = z1_{\set z}$. 
	Similarly, $\comb{D}_0[\comb{I}] \rred \comb{I}1_{\emptyset} 
	\red \esum{\emptyset}$.
\item $(\lam x.\comb{D_0}[x])[z]$ has two redexes. Contracting the outermost first gives $(\lam x.\comb{D_0}[x])[z]\rred \comb{D_0}[z]\rred z1_{\set z}$.
Contracting the innermost $(\lam x.\comb{D_0}[x])[z]\rred (\lam x.x1_{\set x})[z]\rred z1_{\set z}$.
\item $\comb{D}_1[\comb{I},\comb{I}]\rred \comb{I}[\comb{I}] + \comb{I}[\comb{I}]\rred \comb{I}+ \comb{I}[\comb{I}]\rred \comb{I}+\comb{I} = 2.\comb{I}$. 
Thus, sums can arise from single terms.
\esub
\end{exa}

We now show that $\rred$ enjoys the properties of 
strong normalization and strong confluence,
which are the key features of such a resource calculus.

\begin{thm} \label{the:rred-norm-confl}\mbox{}
\bsub
\item The reduction $\rred$ is strongly normalizing.
	\label{the:rred-norm-confl:norm}
\item $\rred$ is strongly confluent in the following sense:
	for all $\rs s, \rs t_1, \rs t_2 \in \frsums[X]**$
	such that $\rs s \rred \rs t_1$ and $\rs s \rred \rs t_2$,
	there exists $\rs u \in \frsums[X]**$ such that
	$\rs t_1 \rredr \rs u$ and $\rs t_2 \rredr \rs u$.
	\label{the:rred-norm-confl:confl}
\esub 
\end{thm}

	\begin{proof}[Proof of \cref{the:rred-norm-confl}, \cref{the:rred-norm-confl:norm}]
	The size $\size s\in\Nat$ of a resource expression $s\in\rterms**$ 
	is defined by structural induction as usual, 
	with base cases $\size{x} = 1$ and $\size{1_X} = 0$, 
	so that $\size{1_X\cdot \rb t} = \size{\rb t}$.
	The \emph{sum-size} of a resource sum $\rs s$ is 
	the finite multiset defined by $\sumsize{\esum{X}} = []$ and 
	$\sumsize{s+\rs t} = [\size{s}]\cdot \sumsize{\rs t}$.
	Sum-sizes are well-ordered by the usual well-founded ordering 
	$\msetlt$ on finite multisets over $\Nat$ 
	(see \cite[\S A.6]{Terese}).
	
	Now, assume that $\rs s \rred\rs s'$. 
	Since the contraction of a redex suppresses an abstraction 
	and there is no duplication, 
	we get $\sumsize{\rs s}\msetlt \sumsize{\rs s'}$. 
	Conclude since $\msetlt$ is well-founded.
	\end{proof}

We now prove \cref{the:rred-norm-confl:confl} of the theorem.
The proof follows a well-trodden path of lemmas \cite{Vaux.19,Cerda.24}.

\begin{lem} \label{lem:rred-and-subst-1}
	If $\rt s \rred \rs s'$ then 
	$\rt s \rsubst t \rredr \rs s' \rsubst t$.
\end{lem}
	
	\begin{proof}
	By induction on a derivation of $\rt s \rred \rs s'$.
	For the base case $(λy.\rt u)\rb v \rred \rt u \rsubst v [y]$,
	\begin{align*}
		\left( (λy.\rt u)\rb v \right) \rsubst t
		& = \sum_{\rb t = \rb t' \cdot \rb t''}
			\left( λy.\rt u \rsubst* {\rb t'} \right)
			\left(\rb v \rsubst* {\rb t''} [y]\right) \\
		& \rredr \sum_{\rb t = \rb t' \cdot \rb t''}
			\rt u \rsubst* {\rb t'}
			\rsubst*{ \rb v \rsubst* {\rb t''} } [y] \\
		& = u \rsubst v [y] \rsubst t,
			& \text{by \cref{lem:resource-substitution}.}
	\end{align*}
	Notice that the reflexive closure $\rredr$ of $\rred$
	is necessary to handle the case where the two sums
	are empty.
	The other cases are immediate applications 
	of the induction hypotheses.
	\end{proof}

\begin{lem} \label{lem:rred-and-subst-2}
	If $\rb t \rred \rbs t'$ then 
	$\rt s \rsubst t \rredr \rt s \rsubst* {\rbs t'}$.
\end{lem}
	
	\begin{proof}
	By a straightforward induction on the different cases
	defining $\rt s \rsubst t$.
	\end{proof}

\begin{lem} \label{lem:rred-confl-one-term}
	For all $\rt s \in \rterms**(X)$ and 
	$\rs s', \rs s'' \in \frsums[X]**$ such that 
	\[ \rt s \rred \rs s' \quad \text{and} \quad \rt s \rred \rs s'', \]
	there exists $\rs t \in \frsums[X]**$ such that
	\[ \rs s' \rredr \rs t \quad \text{and} \quad \rs s'' \rredr \rs t. \]
\end{lem}
	
	\begin{proof}
	Take $\rt s$, $\rs s'$ and $\rs s''$ as above.
	We proceed by induction on both reductions
	$\rt s \rred \rs s'$ and $\rt s \rred \rs s''$.
	The base case is when the first reduction
	comes from the first rule in \cref{def:rred},
	\ie $\rt s = (λx.\rt u)\rb v$
	and $\rs s' = \rt u \rsubst v$:
	\begin{itemize}
	\item If the second reduction comes from the same rule,
		then $\rs s'' = \rs s'$ and we set $\rs t \eqdef \rs s'$.
	\item If the second reduction comes from the rule for left application,
		\ie $\rs s'' = (λx.\rs u'')\rb v$
		with $\rt u \rred \rs u''$, then
		by \cref{lem:rred-and-subst-1}
		$\rt s' = \rt u \rsubst v \rredr \rs u'' \rsubst v$
		and by \cref{def:rred} $\rt s'' \rredr \rs u'' \rsubst v$.
		(The reflexive closure is needed in the latter case
		as $\rs u''$ may be empty.)
	\item If the second reduction comes from the rule for right application,
		then the proof is analogous to the previous case,
		using \cref{lem:rred-and-subst-2}
		instead of \cref{lem:rred-and-subst-1}.
	\end{itemize}
	In all other cases, the result immediately follows from
	the induction hypotheses.
	\end{proof}

	\begin{proof}[Proof of \cref{the:rred-norm-confl}, 
	\cref{the:rred-norm-confl:confl}]
	We prove the result under the slightly more general hypothesis
	that $\rs s \rredr \rs s'$ and $\rs s \rredr \rs s''$.
	We proceed by induction on $\sumsize{\rs s}$,
	as defined in the proof of \cref{the:rred-norm-confl},
	\cref{the:rred-norm-confl:norm} above.
	If $\sumsize{\rs s} = []$, \ie $\rs s$ is empty,
	then both reductions must be equalities
	and the conclusion is trivial.
	Otherwise, suppose $\sumsize{\rs s} \msetgt []$.
	Again, if any of the two reductions
	turns out to be an equality, the result is immediate.
	Otherwise $\rs s \rred \rs s'$ and $\rs s \rred \rs s''$, \ie
	\begin{itemize}
	\item $\rs s = \rt s_1 + \rs s_{\neq 1}$
		and $\rs s' = \rs s'_1 + \rs s'_{\neq 1}$,
		together with $\rt s_1 \rred \rs s'_1$
		and $\rs s_{\neq 1} \rredr \rs s'_{\neq 1}$,
	\item $\rs s = \rt s_2 + \rs s_{\neq 2}$
		and $\rs s'' = \rs s''_2 + \rs s''_{\neq 2}$,
		together with $\rt s_2 \rred \rs s''_2$
		and $\rs s_{\neq 2} \rredr \rs s''_{\neq 2}$.
	\end{itemize}
	There are two possible cases.
	\begin{itemize}
	\item If $\rt s_1 = \rt s_2$ and $\rs s_{\neq 1} = \rs s_{\neq 2}$,
		then by \cref{lem:rred-confl-one-term} applied to
		$\rt s_1 \rred \rs s'_1$ and $\rt s_1 \rred \rs s''_2$,
		we obtain $\rs t_1$ such that
		$\rs s'_1 \rred \rs t_1$ and $\rs s''_2 \rred \rs t_1$.
		By the induction hypothesis on $\rs s_{\neq 1}$,
		we obtain $\rs t_{\neq 1}$ such that
		$\rs s'_{\neq 1} \rred \rs t_{\neq 1}$ 
		and $\rs s''_{\neq 2} \rred \rs t_{\neq 1}$.
		Therefore we can set $\rs t \eqdef \rs t_1 + \rs t_{\neq 1}$.
	\item If $\rt s_1 \neq \rt s_2$, then we can write
		$\rs s = \rt s_1 + \rt s_2 + \rs s_{\neq 12}$, with:
		\begin{align*}
			\rs s_{\neq 1} & = \rt s_2 + \rs s_{\neq 12}, &
			\rs s'_{\neq 1} & = \rs s'_2 + \rs s'_{\neq 12}, &
			\rt s_2 & \rredr \rs s'_2, &
			\rs s_{\neq 12} & \rredr \rs s'_{\neq 12}, \\
			\rs s_{\neq 2} & = \rt s_1 + \rs s_{\neq 12}, &
			\rs s''_{\neq 2} & = \rs s''_1 + \rs s''_{\neq 12}, &
			\rt s_1 & \rredr \rs s''_1, &
			\rs s_{\neq 12} & \rredr \rs s''_{\neq 12}.
		\end{align*}
		By \cref{lem:rred-confl-one-term} applied to $\rt s_1$
		we obtain $\rs t_1$ such that
		$\rs s'_1 \rredr \rs t_1$ and $\rs s''_1 \rredr \rs t_1$.
		By the same lemma applied to $\rt s_2$
		we obtain $\rs t_2$ such that
		$\rs s'_2 \rredr \rs t_2$ and $\rs s''_2 \rredr \rs t_2$.
		By the induction hypothesis on $\rs s_{\neq 12}$
		we obtain $\rs t_{\neq 12}$ such that
		$\rs s'_{\neq 12} \rredr \rs t_{\neq 12}$
		and $\rs s''_{\neq 12} \rredr \rs t_{\neq 12}$.
		Finally, we can set $\rs t \eqdef 
		\rs t_1 + \rs t_2 + \rs t_{\neq 12}$.
	\qedhere
	\end{itemize}
	\end{proof}

By \cref{the:rred-norm-confl} every $\rs s \in \frsums[X]**$ has 
a unique $\mathrm{r}$-normal form, which is denoted by $\nf(\rs s)$. 


\section{Taylor approximation for Ohana trees}\label{sec:TE4MTs}

In its original formulation, Ehrhard and Regnier's Taylor expansion 
translates a \lam-term as a power series of iterated 
derivatives~\cite{Ehrhard.Reg.08}.
We now introduce a \emph{qualitative} Taylor expansion~\cite{ManzonettoP11} specifically designed for the \lamI-calculus, having as target the \lamI-resource calculus.

\subsection{The Taylor approximation}

Intuitively, a qualitative Taylor expansion should associate each term $M \in \LamI(X)$ with a \emph{set} of resource approximants $\TE M \subseteq \rterms(X)$.
Therefore the codomain of $\TE -$ should be $\bigdunion_{X \subf \Var} \parts{\rterms(X)}$, and what we actually define is $\TE M \eqdef (X, \TE* M)$ with $\TE* M \subseteq \rterms(X)$.
Notice that, whereas all previous constructions of disjoint unions
$\bigdunion_{X \subf \Var}$ were formed from genuinely disjoint sets,
this is no longer the case here, as $\emptyset \in \parts{\rterms(X)}$, for all $X$.
\begin{nota}\mbox{}
\bsub
\item For $X \subf \Var$ and $\cX \in \parts{\rterms(X)}$,
	we write $\fv(X,\cX) \eqdef X$.
\item We let $\mathord{\TElt}$ denote the order relation
	on $\bigdunion_{X \subf \Var} \parts{\rterms(X)}$
	such that $(X,\cX) \TElt (Y,\cY)$
	whenever $X = Y$ and $\cX \subseteq \cY$.
	We write $\TElub$ for the corresponding least upper bounds.
\item We write $\rt s \TEin (X,\cX)$ to mean that $\rt s \in \cX$.
\item Given $\cX \in \parts{\rterms(X)}$, 
	$\cY \in \parts{\rterms*(Y)} \setminus \set\emptyset $
	and $x \in X$, we write:
	\begin{gather*}
		\lam x.(X,\cX) \eqdef (X \setminus \set x, \set{\lam x.\rt s}[\rt 
		s\in\cX]), \\
		(X,\cX) \cY \eqdef (X \cup Y, \set{\rt s\rb t}[\rt s \in \cX,\ \rb t \in \cY]). 
	\end{gather*}
\esub
\end{nota}

\begin{defi}\mbox{} \label{def:taylor}
\bsub
\item The \defemph{Taylor expansion} $\TE{M} \in \bigdunion_{X \subf \Var} \parts{\rterms(X)}$ of a \lamI-term $M$, is defined together with $\TE** M \subset \rterms*(\fv(M))$ by mutual induction:
	\begin{gather*}
		\begin{aligned}
		\TE{x} & \eqdef (\set{x}, \set{x}), &
		\TE{\lam x.M} & \eqdef \lam x.\TE M, &
		\TE{MN} &\eqdef  \TE M \TE** N, 
		\end{aligned} \\
		\TE**{M} \eqdef \set{1_{\fv(M)}} \cup 
			\set{ [t_0,\dots,t_n] } [ n\in\Nat,\ t_0,\dots,t_n\TEin\TE M ]. 
	\end{gather*}
\item The above definition is extended to Ohana trees by setting,
	for all $M \in \LamI$,
	\[	\TMT M \eqdef \bigTElub_{A \in \Appof{M}} \TE{A},
		\qquad \text{together with }
		\left\{ \begin{array}{r!{\eqdef}l}
			\TE{\bot_X} & (X,\emptyset), \\ 
			\TE**{\bot_X} & \set{1_X}.\\
		\end{array} \right.
	\]
\esub
\end{defi}

\begin{rem} \label{rem:fv-of-taylor}
	For all $M$, we have $\fv(\TE M) = \fv(M)$.
	Similarly, $\fv(\TMT M) = \fv(M)$,
	due to \cref{rem:appof-M-fv-M} and the way
	we ordered $\bigdunion_{X \subf \Var} \parts{\rterms(X)}$.
\end{rem}

\begin{exa}\mbox{}
	\bsub
	\item $\TE{\comb I} = (\emptyset, \set{\comb I})$,
		$\TE{\comb D} = (\emptyset, \set{\comb D_n}[n \in \Nat])$.
	\item $\TMT{\comb Yf} = (\set f, \cX_{\comb Yf})$ and
		$\TMT{\Bible[l]f} = (\set{f,l}, \cX_{\Bible[l]f})$,
		where the sets of approximants can be described 
		as the smallest (in fact, unique) subsets of $\rterms$ such that:
		\begin{align*}
			\cX_{\comb Yf} & = \set{ f1_{\set f} } \cup 
				\set{ f[\rt t_0,\dots,\rt t_n] }
				[n\in\Nat,\ \rt t_0,\dots,\rt t_n \in \cX_{\comb Yf} ], \\
			\cX_{\Bible[l]f} & = \set{ f1_{\set{f,l}} } \cup 
				\set{ f[\rt t_0,\dots,\rt t_n] }
				[n\in\Nat,\ \rt t_0,\dots,\rt t_n \in \cX_{\Bible[l]f} ].
		\end{align*}
	\esub
\end{exa}

We now describe how resource reduction acts on Taylor expansions, \ie on potentially infinite sets of resource expressions.

\begin{nota} \label{not:support}\mbox{}
\bsub
\item Given $\rs s \in \frsums[X]**$,
	we denote by $\supp s \in \bigdunion_{X \subf \Var} \parts{\rterms(X)}$ 
	its \defemph{support}, defined by
	$\supp*{\esum X} \eqdef (X,\emptyset)$ and
	$\supp*{\rt s + \rs t} \eqdef (X,\set{\rt s}) \TElub \supp t$.
	Notice that $\rt s \TEin \supp t$ whenever $\rt s$ 
	bears a non-zero coefficient in $\rs t$ (\ie$\rs t = s + \rs t'$, for some $\rs t'$).
\item Given $\cX \in \parts{\rterms**(X)}$, we write 
	$\nf(X, \cX) \eqdef \bigTElub_{\rt s \in \cX} \supp*{ \nf(\rt s) }$.
	Notice that $\fv(\nf(X, \cX)) = X$, and that
	$\nf(X, \cX) = (X, \emptyset)$ if and only if
	$\rt s \rreds \esum X$, for all $\rt s \in \cX$.
\esub
\end{nota}

This allows us to state the following theorem,
adapting \cite[Theorem~2]{Ehrhard.Reg.06}.

\begin{thm}[Commutation] \label{the:commutation}
	For all $M \in \LamI$, $\NFT{M} = \TMT{M}$.
\end{thm}

We outline a proof in the following sequence of lemmas,
similar to the one in \cite{Barbarossa.Man.19};
alternatively, the variants of \cite{Ehrhard.Reg.06,Ehrhard.Reg.08},
\cite{OlimpieriVaux22} and \cite{Cerda.24,CerdaVaux23}
could also be adapted.

\newcounter{lem:rred-conserv-bred}
\setcounter{lem:rred-conserv-bred}{\arabic{thm}}
\begin{lem}
\label{lem:rred-conserv-bred}
	For all resource sums $\rs s$ such that $\supp s \TElt \TE{M}$,
	there exists a reduction $M \breds N$
	such that $\supp*{\nf(\rs s)} \TElt \TE{N}$.
\end{lem}

	\begin{proof}
	By induction on the length of the longest reduction 
	$\rs s \rreds \nf(\rs s)$.
	If it is $0$, take $N \eqdef M$.
	Otherwise, the longest reduction is
	$\rs s = \rt t + \rs u \rred \rs t' + \rs u \rreds \nf(\rs s)$.
	By firing all redexes in $M$ and $\rs u$ at the same position
	as the redex fired in $\rt t \rred \rs t'$
	(formally we do this by induction on this reduction),
	we obtain $M \bred M'$ and $\rs u \rreds \rs u'$ such that
	$\supp*{\rs t' + \rs u'} \TElt \TE{M'}$.
	By confluence, $\rs t' + \rs u' \rreds \nf(\rs s)$.
	We conclude by induction on this (shorter) reduction.
	\end{proof}

\newcounter{lem:normal-approximants}
\setcounter{lem:normal-approximants}{\arabic{thm}}
\begin{lem}
\label{lem:normal-approximants}
	If $\rt t \TEin \TE N$ is in $\resource$-nf, then
	there exists $A \in \Appof{N}$ such that $\rt t \TEin \TE A$.
\end{lem}

	\begin{proof}
	By induction on $t$. An $\resource$-nf is also a hnf, hence
	$\rt t = \lam \vec{x}.y\rb t_1\cdots\rb t_k$ and
	$N = \lam \vec{x}.yN_1\cdots N_k$.
	For $1 \leq i \leq k$, $\rb t_i \in \TE**{N_i}$.
	If $\rb t_i = 1_{\fv(N_i)}$ define $A_i \eqdef \bot_{\fv(N_i)}$.
	If $\rb t_i = [t_i^1,\dots,t_i^n]$,
	each $t_i^j$ is in $\resource$-nf, by induction
	there is $A_i^j \in \Appof{N_i}$, $t_i^j \TEin \TE{A_i^j}$.
	Define $A_i \eqdef \bigMTlub_{j=1}^n A_i^j \in \Appof{N_i}$.
	Finally, $A \eqdef \lam \vec{x}.yA_1\cdots A_k \in \Appof{N}$
	and $t \TEin \TE A$.
	\end{proof}

\newcounter{lem:TE-monotone}
\setcounter{lem:TE-monotone}{\arabic{thm}}
\begin{lem}[Monotonicity of $\TE{-}$]\mbox{}
\label{lem:TE-monotone}
	\bsub
	\item $\TE{A} \TElt \TE{A'}$ if and only if $A \MTle A'$.
	\item For all $N \in \LamI$, we have $\TE{\dirapp N} \TElt \TE N$.
	\esub
\end{lem}

	\begin{proof}
	Immediate inductions (i) on $A$ and $A'$,
	and (ii) on the head structure of $N$.
	\end{proof}

\newcounter{lem:rred-simul-bred}
\setcounter{lem:rred-simul-bred}{\arabic{thm}}
\begin{lem}[Simulation of $\bred$]
\label{lem:rred-simul-bred}
	If $M \bred N$, then $\TE{N} = 
	\bigTElub_{\rt s \TEin \TE M} \supp*{\rs t_{\rt s}}$
	for resource sums $\rs t_{\rt s} \in \frsums[\fv(M)]$
	such that $\forall \rt s \TEin \TE M,\ \rt s \rreds \rs t_{\rt s}$.
\end{lem}

	\begin{proof}
	One first shows by induction on $P$ that
	for all $P,Q \in \LamI$ and $x \in \fv(P)$,
	$\TE{P \subst Q} = \bigTElub_{\rt s \TEin \TE P}
	\bigTElub_{\rb t \in \TE** Q} \supp*{ \rt s \rsubst t }$.
	Then the proof is an induction on a derivation of $M \bred N$,
	using the previous equality in the base case.
	See the details of a similar proof in 
	\cite[Lemmas 4.1 and 4.2]{CerdaVaux23}.
	\end{proof}

\newcounter{lem:TE-of-A-is-normal}
\setcounter{lem:TE-of-A-is-normal}{\arabic{thm}}
\begin{lem}
\label{lem:TE-of-A-is-normal}
	If $A \in \Appset$ then all $\rt s \TEin \TE A$
	are in $\resource$-nf.
\end{lem}

	\begin{proof}
	Immediate induction on $A$.
	\end{proof}

Everything is now in place to prove the Commutation \cref{the:commutation}.

\begin{proof}[Proof of \cref{the:commutation}]
	By \cref{rem:fv-of-taylor}, $\fv(\NFT M) = \fv(M) = \fv(\TMT M)$.
	Thus, it is sufficient to prove that, 
	for $\rt t \in \rterms(\fv(M))$:
	$\rt t \TEin \NFT M$ if and only if $\rt t \TEin \TMT M$.
	
	Take $\rt t \TEin \NFT{M}$, \ie $\rt t \in \supp*{\nf(\rt s)}$
	for some $\rt s \TEin \TE M$.
	Then, by \cref{lem:rred-conserv-bred}, there exists
	a reduction $M \breds N$
	such that $\supp*{\nf(\rt s)} \TElt \TE{N}$.
	Hence $\rt t \TEin \TE{N}$ and $\rt t$ is in $\resource$-nf: 
	\cref{lem:normal-approximants} ensures that $\rt t \TEin \TE A$ 
	for some $A\in \Appof{N}=\Appof{M}$ (by \cref{prop:InvariantUnderBeta}).
	This means exactly that $\rt t \TEin \TMT M$.
	
	Conversely, take $\rt t \TEin \TMT M$, \ie
	$\rt t \TEin \TE{A}$ for some $A \in \Appof{M}$.
	By definition there is a reduction $M \breds N$
	such that $A \MTle \dirapp{N}$.
	By \cref{lem:TE-monotone},
	\[\rt t \TEin \TE A \TElt \TE{\dirapp N} \TElt \TE N.\]
	By iterated applications of \cref{lem:rred-simul-bred}, 
	there is an $\rt s \TEin \TE M$ such that
	$\rt s \rreds \rs t_{\rt s}$ and $\rt t \in \supp*{\rs t_{\rt s}}$.
	By \cref{lem:TE-of-A-is-normal}, $\rt t$ is in $\resource$-nf,
	therefore $\rt t \TEin \supp*{\nf(\rt s)} \TElt \NFT M$.
\end{proof}

\subsection{The \lamI-theory of Ohana trees}

Consider the equivalence $\cM$ on $\LamI$, 
defined by $M =_{\cM} N$ if and only if $\MT{M} = \MT{N}$.
Thanks to \cref{the:commutation},
we are now able to show that this equivalence is a \lamI-theory.

\begin{cor} \label{cor:MTeqIffNFTeq}
	For $M,N\in\LamI$, $M =_\cM N$ if and only if $\NFT M = \NFT N$.
\end{cor}
	\begin{proof}
	($\Rightarrow$)
	Immediate by \cref{the:commutation}.
	
	($\Leftarrow$)
	Take $M,N$ such that $\NFT M = \NFT N$. 
	By \cref{thm:MTbijApp}, it is sufficient to prove 
	$\Appof{M} = \Appof{N}$.
	Take $A\in\Appof{M}$. 
	If $A = \bot_X$, then $A \in \Appof{N}$ is trivial 
		since $\Appof{N}$ is downward closed.
	Otherwise, by hypothesis and \cref{the:commutation},
		$\TE A \TElt \bigTElub_{B \in \Appof N} \TE  B$.
		There is a $B \in \Appof N$ such that
		$\resourceof A \TEin \TE B$,
		where $\resourceof A \TEin \TE A$ is defined by
		\[
		\begin{array}{rcl}
		\resourceof x &\eqdef& x,\\
\resourceof{\lam \vec{x}.yA_1\cdots A_k} &\eqdef&
			\lam \vec{x}.y\resourceof*{A_1}\cdots \resourceof*{A_k},\\	
		\end{array}
		\qquad
		\resourceof*{A} \eqdef
		\begin{cases}
		1_X,&\textrm{if }A = \bot_X,\\
		{[\resourceof A]},&\textrm{otherwise}.
		\end{cases}
		\]
		By induction on $A$, one shows that
		$\resourceof A \TEin \TE B$ implies $A \MTle B$,
		hence $A \in \Appof N$ by downward-closure
		(\cref{lem:DirApp}\ref{lem:DirApp2}).
	This shows that $\Appof{M} \subseteq \Appof{N}$.
	The converse inclusion is symmetric.
	\end{proof}

\begin{cor}\label{cor:MisALamITheory} 
	$\cM$ is a \lamI-theory.
\end{cor}

\begin{proof}
	The relation $=_\cM$ is clearly an equivalence.
	By \cref{prop:InvariantUnderBeta,thm:MTbijApp}, 
	for all $M,N\in\LamI$, $M =_\beta N$ entails $M =_\cM N$.
	
	For compatibility with abstraction, 
	take $M =_\cM N$ and $x \in \fv(M) \cap \fv(N)$. 
	By \cref{cor:MTeqIffNFTeq}, we get 
	$\NFT{\lam x.M} = \lam x.\NFT M = \lam x.\NFT N = \NFT{\lam x.N}$.
	
	For compatibility with application, 
	by \cref{the:rred-norm-confl} observe that for all $M,N \in \LamI$, 
	$\NFT{MN} = \nf(\TE M \TE** N) = \nf(\NFT M \NFT** N)$,
	where $\NFT** N$ is defined by adapting
	\cref{not:support} to sets of resource bags.
	Therefore, $\NFT M = \NFT{M'}$ and $\NFT N = \NFT{N'}$ 
	imply $\NFT{MN} = \NFT{M'N'}$, and
	we can conclude the proof by \cref{cor:MTeqIffNFTeq}.
\end{proof}


\section{A multi-type system characterizing $\cO$}
\label{sec:types}

Once established that Ohana trees induce a \lamI-theory, the question naturally arises whether it is possible to design a denotational model whose theory is exactly~$\cO$. In the conference version of this work \autocite{Cerda.Man.Sau.25}, this was presented as Problem~46.
In fact, this problem immediately leads to a further difficulty: 
since $\cO$ is not a \lam-theory
(as explained in \cref{sec:Ohana-lamK}),
it cannot be the theory of a denotational model of the \lam-calculus,
but only of a denotational model of the \lamI-calculus ---
which is not a particularly well-defined notion.
Although some proposals exist \autocite{EgidiHR92,Jacobs93}, unifying them and ensuring their compatibility with the present framework is left for future work.

In this section we present a first step towards such a denotational model by introducing a multi-type (or non-idempotent intersection type) system characterizing $\cO$ just like the usual system characterizes the theory $\cB$ induced by equality of Böhm trees
\autocite{Rocca82,BreuvartMR18}:
\[	\forall M,N \in \Lam,\quad
	\BT{M} = \BT{N} \Longleftrightarrow
	\set{(\Gamma,\sigma)}[\Gamma \vdash M : \sigma]
		= \set{(\Gamma,\sigma)}[\Gamma \vdash N : \sigma]. \]
By interpreting the set
$\interp{M} \eqdef \set{(\Gamma,\sigma)}[\Gamma \vdash M : \sigma]$
in the category $\mathbf{Rel}$ of sets and relations
this usually gives rise to
the relational model of the \lam-calculus
\autocite[Thm.~6.3.15~sqq.]{deCarvalho.07};
the current section paves the way for a similar construction
for the \lamI-calculus.

\subsection{A multi-type system}

Let us first introduce a grammar generating the multi-types.
We adapt the standard definition along the same lines as in the treatment of the resource calculus and the Taylor expansion in \cref{def:rterms,def:taylor}:
empty multisets $[]_X$ are annotated with a finite set $X$ of variables, and the types of a term are extended with a type $\esum{X}$ playing the role of the first component of the Taylor expansion.

\begin{defi}
	Given a countable set $\Atoms$ of \emph{atoms},
	the set $\OTypes$ of \defemph{multi-types} is defined inductively by:
	\[
		\begin{array}{r@{\quad\ni\quad}r@{\quad\eqbnf\quad}l@{\qquad}r}
		\OTypes & \sigma,\tau,\dots & \esum X \isep \alpha & X \subf \Vars 
			\\[\topsep]
		\Types & \alpha,\beta,\dots & \atom \isep \rb\alpha \lto \beta
			& \atom \in \Atoms \\[\topsep]
		\MTypes & \rb\alpha, \rb\beta, \dots & []_X \isep [α_0,\dots,α_n]
			& X \subf \Vars,\ n \in \Nat
		\end{array}
	\]
\end{defi}

The key feature of our typing system is that, along the usual environment $\Gamma$ associating a multiset of types to all the free variables of the typed term,
we introduce a second environment $\Delta$ in each sequent:
to each free variable $x$ of the typed term $M$ this environment associates \emph{the set of the free variables of the term that will ultimately be substituted to $x$ in $M$}.
This should become clear in the statement and the proof of the substitution \cref{lem:typing:fw-subst-lemma-empty,%
lem:typing:fw-subst-lemma-nonempty,lem:typing:bw-subst-lemma} below.

\begin{defi}
	An \defemph{environment} is a pair denoted by $\Gamma \ctxsep \Delta$,
	where:
	\begin{itemize}
	\item A \emph{(type) environment} $\Gamma$ is a map $\Vars \to \Mfin{\Types}$ whose support $\support(\Gamma) \eqdef \set{ x \mid \Gamma(x) \neq []}$ is finite.
	In other words, we have that $\Gamma(x) \neq []$ holds for finitely many variables $x$.
		\begin{itemize}
		\item Such a $\Gamma$ can be presented as a sequence $x_1:\rb\alpha_1, \dots, x_n:\rb\alpha_n$ where the $x_i$ are supposed distinct.
			Observe that the notation $x_i:\rb{\alpha}_i$ is slightly improper here, since these $\rb{\alpha}_i$ live in $\Mfin{\Types}$, not in $\MTypes$.
			Concretely, this means that when $\rb{\alpha} = []_X\in \MTypes$ is used in an environment $\Gamma, x:\rb\alpha$, we silently forget the annotation $X$ and set $\rb\alpha = []$.
		\item Such environments are equipped with an operation $+$ performing pointwise multiset union: 
		\[
			(\Gamma_1+\Gamma_2)(x) = \Gamma_1(x) + \Gamma_2(x),\textrm{ for all $x\in\Var$.}
		\]
		\end{itemize}
	\item A \emph{(variable) environment} is any finite partial map $\Delta : \Vars \pto \partsfin{\Vars}$. 
		\begin{itemize}
		\item Such a $\Delta$ is presented as a sequence $x_1:X_1,\dots,x_n:X_n$ where the $x_i$ are implicitly distinct and $X_i \in \partsfin{\Vars}$, for all $1\le i \le n$.
		\item We denote by $\domain(\Delta)$ the domain of $\Delta$,
			and we define $\uimage(\Delta) \eqdef \bigcup_{i=1}^n X_i$.
		\item We write $\Delta_1 \coh \Delta_2$ if, for all $x\in \domain(\Delta_1) \cap \domain(\Delta_2)$, we have $\Delta_1(x) = \Delta_2(x)$. 
		\item For $\Delta_1,\Delta_2$ such that $\Delta_1\coh \Delta_2$, we define
			\[
			(\Delta_1\vee \Delta_2)(x) \eqdef 
				\begin{cases}
				\Delta_1(x),& \text{if } x \in \domain(\Delta_1),\\
				\Delta_2(x),&\text{otherwise}.
				\end{cases}
			\]
		\item We write $\Delta_1 \subseteq \Delta_2$ whenever $\domain(\Delta_1) \subseteq \domain(\Delta_2)$,
			and $\Delta_1$ and $\Delta_2$ coincide on $\domain(\Delta_1)$.
			Observe that $\Delta_1 \subseteq \Delta_2 \coh \Delta_3$ entails $\Delta_1 \coh \Delta_3$.
		\end{itemize}
	\end{itemize}
\end{defi}

\begin{defi} \label{def:typing-rules}
	We define the typing judgements 
	 $\Gamma\ctxsep \Delta \vdash M : \sigma$
	and
	 $\Gamma\ctxsep \Delta \vdashbang M :\rb\alpha$
	by mutual induction as follows:
	\begin{gather*}
		\begin{prooftree}
		\hypo{ \domain(\Delta) = \fv(M) }
		\infer1[\esum{}]{ \ctxsep \Delta \vdash M : \esum{\uimage(\Delta)} }
		\end{prooftree}
	\\[\topsep]
		\begin{prooftree}
		\infer0[ax]{ x : [\alpha] \ctxsep x : X \vdash x : \alpha }
		\end{prooftree}
	\qquad
		\begin{prooftree}
		\hypo{ \Gamma_0 \ctxsep \Delta_0 \vdash M : \rb\alpha \lto \beta }
		\hypo{ \Gamma_1 \ctxsep \Delta_1 \vdashbang N : \rb\alpha }
		\hypo{ \Delta_0 \coh \Delta_1 }
		\infer3[@]{ \Gamma_0 + \Gamma_1 \ctxsep \Delta_0 \vee \Delta_1
			\vdash MN : \beta }
		\end{prooftree}
	\\[\topsep]
		\begin{prooftree}
		\hypo{ \Gamma, x : [] \ctxsep \Delta, x : X \vdash M : \beta }
		\infer1[λ^0]{ \Gamma \ctxsep \Delta \vdash λx.M : []_X \lto \beta }
		\end{prooftree}
	\qquad
		\begin{prooftree}
		\hypo{ \Gamma, x : [\alpha_0, \dots, \alpha_n] \ctxsep \Delta, x : X 
			\vdash M : \beta }
		\infer1[λ^+]{ \Gamma \ctxsep \Delta 
			\vdash λx.M : [\alpha_0, \dots, \alpha_n] \lto \beta }
		\end{prooftree}
	\\[\topsep]
		\begin{prooftree}
		\hypo{ \domain(\Delta) = \fv(M) }
		\infer1[!^0]{ \ctxsep \Delta \vdashbang M : []_{\uimage(\Delta)} }
		\end{prooftree}
	\qquad
		\begin{prooftree}
		\hypo{ \Gamma_0 \ctxsep \Delta \vdash M : \alpha_0 }
		\hypo{ \cdots }
		\hypo{ \Gamma_n \ctxsep \Delta \vdash M : \alpha_n }
		\infer3[!^+]{ \Gamma_0 + \cdots + \Gamma_n \ctxsep \Delta \vdashbang 
			M : [\alpha_0, \dots, \alpha_n] }
		\end{prooftree}
	\end{gather*}
	We write $\pi \derives \Gamma \ctxsep \Delta \vdash M : \sigma$ whenever $\pi$ is a derivation of the given conclusion, and $\Gamma \ctxsep \Delta \vdash M : \sigma$ to express that such a derivation exists.
\end{defi}

\emph{A priori} this typing system only acts at the level of \enquote{raw} \lamI-terms, \ie terms not quotiented by α-equivalence.
However, it should be clear that typing derivations can be endowed with an α-equivalence relation such that if $M =_α M'$ and $\pi \derives \Gamma \ctxsep \Delta \vdash M : \sigma$,
then there is a $\pi' =_α \pi$
such that $\pi' \derives \Gamma \ctxsep \Delta \vdash M' : \sigma$.
As a consequence we silently quotient derivations by α-equivalence and make the typing system act on (quotiented) \lamI-terms.

The two main results displaying the good behaviour
of any reasonable multi-type system,
which we shall now prove,
are subject reduction and subject expansion
(\cref{thm:typing:subject-reduction,thm:typing:subject-expansion}).

\begin{lem} \label{lem:typing:domain-subset-domain-eq-fv}
	If there is a derivation $\Gamma \ctxsep \Delta \vdash M : \alpha$
	or $\Gamma \ctxsep \Delta \vdashbang M : \rb\alpha$,
	then $\support(\Gamma) \subseteq \domain(\Delta) = \fv(M)$.
\end{lem}
	
	\begin{proof}
	Immediate induction.
	\end{proof}

\begin{lem}[Forward Substitution Lemma, case~1]
\label{lem:typing:fw-subst-lemma-empty}
	If there are derivations
	$\Gamma \ctxsep \Delta, x:\uimage(\Delta') \vdash M : \beta$
	and $\ctxsep \Delta' \vdashbang N : []_{\uimage(\Delta')}$
	such that $x \notin \support(\Gamma)$ and $\Delta \coh \Delta'$,
	then there is a derivation
	$\Gamma \ctxsep \Delta \vee \Delta' \vdash M \subst N : \beta$.
\end{lem}
	
	\begin{proof}
	We proceed by induction on the derivation of
	$\Gamma \ctxsep \Delta, x:\uimage(\Delta') \vdash M : \beta$.
	
	\begin{itemize}
	\item The last rule cannot be \drule{ax} because $\Gamma$ and
		$\Delta, x:\uimage(\Delta')$ do not have identical domains.
		
	\item If the last rule is \drule{λ^0}
		then $M = λy.P$ with $y \neq x$ and $y \notin \fv(N)$,
		$β = []_Y \lto \delta$, and the rule's hypothesis is a derivation:
		\[	\Gamma \ctxsep \Delta, x:\uimage(\Delta'), y:Y
			\vdash P:\delta. \]
		Since $\Delta \coh \Delta'$ and $y \notin \fv(N)$,
		we have $(\Delta, y:Y) \coh \Delta'$,
		and we can apply the induction hypothesis:
		\[	\Gamma \ctxsep (\Delta, y:Y) \vee \Delta'
			\vdash P \subst N : \delta \]
		and conclude by observing that
		$(\Delta, y:Y) \vee \Delta' = \Delta \vee \Delta', y:Y$
		and by applying the rule \drule{λ^0} again.
	
	\item If the last rule is \drule{λ^+} the proof is the same
		as in the previous case.
	
	\item If the last rule is \drule{@} then $M = PQ$,
		and there are two possible cases for the hypotheses of the rule.
		The first one is as follows:
		\[\resizebox{\linewidth}{!}{
			\begin{prooftree}
			\hypo{\vdots}
			\infer1{ \Gamma \ctxsep \Delta_0, x:X_0
				\vdash P : []_Z \lto \beta }
			\hypo{ \domain(\Delta_1, x:X_1) = \fv(Q) }
			\infer1[!^0]{ \ctxsep \Delta_1, x:X_1 \vdashbang Q : []_Z }
			\hypo{ (\Delta_0, x:X_0) \coh (\Delta_1, x:X_1) }
			\infer3[@]{ \Gamma \ctxsep \Delta, x:\uimage(\Delta')
				\vdash PQ : \beta }
			\end{prooftree} 
		} \]
		where $\Delta = \Delta_0 \vee \Delta_1$,
		the sets $X_0$ and $X_1$ are equal either to $\emptyset$
		or to $\uimage(\Delta')$, and cannot be both empty,
		and $Z \eqdef \uimage(\Delta_1, x:X_1)$.
		
		If $X_0 = \emptyset$,
		\ie by \cref{lem:typing:domain-subset-domain-eq-fv}
		if $x \notin \fv(P)$,
		then the first hypothesis of \drule{@} becomes
		$\Gamma \ctxsep \Delta_0 \vdash P \subst N = P : []_Z \lto β$.
		Otherwise if $X_0 = \uimage(\Delta')$ then
		since $\Delta_0 \coh \Delta'$ (because $\Delta \coh \Delta'$)
		we can apply the induction hypothesis and obtain a derivation of
		$\Gamma \ctxsep \Delta_0 \vee \Delta'
		\vdash P \subst N : []_Z \lto β$.
		
		If $X_1 = \emptyset$,
		\ie by \cref{lem:typing:domain-subset-domain-eq-fv}
		if $x \notin \fv(Q)$,
		then the second hypothesis of \drule{@} becomes
		$\ctxsep \Delta_1 \vdashbang Q \subst N = Q : []_Z$.
		Otherwise if $X_1 = \uimage(\Delta')$ then
		since $\Delta_1 \coh \Delta'$ (because $\Delta \coh \Delta'$),
		$\Delta_1 \vee \Delta'$ is defined and verifies:
		\begin{gather*}
			\domain(\Delta_1 \vee \Delta')
			= \domain(\Delta_1, x:X_1) \setminus \set x
				\cup \domain(\Delta')
			= \fv(Q) \setminus \set x \cup \fv(N)
			= \fv(Q \subst N) \\
			\uimage(\Delta_1 \vee \Delta')
			= \uimage(\Delta_1) \cup \uimage(\Delta')
			= Z
		\end{gather*}
		so that the rule \drule{!^0} allows to derive
		$\ctxsep \Delta_1 \vee \Delta' \vdashbang Q \subst N : []_Z$.
			
		In addition, $\Delta_0 \coh \Delta_1$
		(as a consequence of the third hypothesis of \drule{@})
		and $(\Delta_0 \vee \Delta_1) \coh \Delta'$,
		therefore $\Delta_0$, $\Delta_1$, $\Delta_0 \vee \Delta'$
		and $\Delta_1 \vee \Delta'$ are all pairwise coherent.
		
		As a consequence we can apply the rule \drule{@} again
		in each of the three possible cases
		($X_0 = \varnothing$ and $X_1 = \uimage(\Delta')$,
		$X_0 = \uimage(\Delta')$ and $X_1 = \varnothing$, and
		$X_0 = X_1 = \uimage(\Delta')$).
		In each case, by the associativity and the idempotency of $\vee$,
		we obtain $\Gamma \ctxsep \Delta \vee \Delta'
		\vdash (P \subst N)(Q\subst N): β$.
		
		The second case for the hypotheses of \drule{@} is as follows:
		\[
			\resizebox{\linewidth}{!}{	
			\begin{prooftree}
			\hypo{\vdots}
			\infer1{ \Gamma_0 \ctxsep \Delta_0, x:X_0
				\vdash P : \rb\gamma \lto \beta }
			\hypo{\vdots}
			\infer1{ \Gamma_{1,i} \ctxsep \Delta_1, x:X_1
				\vdash Q : \gamma_i }
			\delims{ \left[ }{ \right]_{i=0}^n }
			\infer[left label={!^+}]1{ \Gamma_1 \ctxsep \Delta_1, x:X_1 
				\vdashbang Q : \rb\gamma }
			\hypo{ (\Delta_0, x:X_0) \coh (\Delta_1, x:X_1) }
			\infer3[@]{ \Gamma \ctxsep \Delta, x:\uimage(\Delta')
				\vdash PQ : \beta }
			\end{prooftree} 
			}
		\]
		where $\rb\gamma = [\gamma_0, \dots, \gamma_n]$,
		$\Gamma = \Gamma_0 + \Gamma_1$
		and $\Gamma_1 = \Gamma_{1,1} + \dots + \Gamma_{1,n}$,
		$\Delta = \Delta_0 \vee \Delta_1$,
		and the sets $X_0$ and $X_1$ are equal either to $\emptyset$
		or to $\uimage(\Delta')$ and cannot be both empty.
		
		Again, if $X_0 = \emptyset$,
		\ie by \cref{lem:typing:domain-subset-domain-eq-fv}
		if $x \notin \fv(P)$,
		then the first hypothesis of \drule{@} becomes
		$\Gamma_0 \ctxsep \Delta_0 \vdash P \subst N = P : \rb\gamma \lto β$.
		Otherwise if $X_0 = \uimage(\Delta')$ then
		since $\Delta_0 \coh \Delta'$ (because $\Delta \coh \Delta'$)
		we can apply the induction hypothesis and obtain a derivation of
		$\Gamma_0 \ctxsep \Delta_0 \vee \Delta'
		\vdash P \subst N : \rb\gamma \lto β$.
		We proceed similarly in each of the $n$ hypotheses of \drule{!^+}
		and apply again \drule{!^+}, obtaining
		either $\Gamma_1 \ctxsep \Delta_1
		\vdashbang Q \subst N = Q : \rb\gamma$ if $X_1 = \varnothing$,
		or $\Gamma_1 \ctxsep \Delta_1 \vee \Delta'
		\vdashbang Q \subst N : \rb\gamma$ if $X_1 = \uimage(\Delta')$.
		We conclude as in the first case.
	\qedhere
	\end{itemize}
	\end{proof}

\begin{lem}[Forward Substitution Lemma, case~2]
\label{lem:typing:fw-subst-lemma-nonempty}
	If there are derivations
	$\Gamma, x:\rb\alpha \ctxsep \Delta, x:\uimage(\Delta')
		\vdash M : \beta$
	and $\Gamma' \ctxsep \Delta' \vdashbang N : \rb\alpha$
	such that $\rb\alpha \neq []$ and $\Delta \coh \Delta'$,
	then there is a derivation
	$\Gamma + \Gamma' \ctxsep \Delta \vee \Delta' \vdash M \subst N : \beta$.
\end{lem}
	
	\begin{proof}
	We proceed by induction on the derivation
	$\Gamma \ctxsep \Delta, x:\uimage(\Delta') \vdash M : \beta$.
	
	\begin{itemize}
	\item If the last rule is \drule{ax} then $M = x$,
		$\Gamma$ and $\Delta$ are empty,
		and $\rb{\alpha} = [\beta]$.
		The derivation $\Gamma' \ctxsep \Delta' \vdashbang N : [β]$
		must be obtained by the rule \drule{!^+}
		with the hypothesis $\Gamma' \ctxsep \Delta' \vdash N : β$,
		which is the desired result since $x \subst N = N$.
	
	\item If the last rule is \drule{λ^0} or \drule{λ^+}
		the proof is the same as for \cref{lem:typing:fw-subst-lemma-empty}.
	
	\item If the last rule is \drule{@}, then $M = PQ$
		and just as for \cref{lem:typing:fw-subst-lemma-empty}
		there are two cases for the hypotheses of the rule.
		
		The first case is the same as for 
		\cref{lem:typing:fw-subst-lemma-empty}
		with only two small differences in
		the first hypothesis of \drule{@} which becomes
		$\Gamma, x:\rb{\alpha} \ctxsep
		\Delta_0, x:\uimage(\Delta') \vdash P : []_Z \lto \beta$
		(the case $X_0 = \emptyset$ is not possible any more).
		The proof is exactly as in \cref{lem:typing:fw-subst-lemma-empty},
		in particular as the first difference is exactly what
		one needs to be able to apply the induction hypothesis.
		
		The second case for the hypotheses of \drule{@} is as follows:
		\[	\resizebox{\linewidth}{!}{\begin{prooftree}
			\hypo{\vdots}
			\infer1{ \Gamma_0, x:\rb{\alpha}_0 \ctxsep \Delta_0, x:X_0
				\vdash P : \rb\gamma \lto \beta }
			\hypo{\vdots}
			\infer1{ \Gamma_{1,i}, x:\rb{\alpha}_{1,i}
				\ctxsep \Delta_1, x:X_1	\vdash Q : \gamma_i }
			\delims{ \left[ }{ \right]_{i=0}^n }
			\infer[left label={!^+}]1{ \Gamma_1, x:\rb{\alpha}_1
				\ctxsep \Delta_1, x:X_1 \vdashbang Q : \rb\gamma }
			\hypo{ (\dots) }
			\infer3[@]{ \Gamma, x:\rb{\alpha} 
				\ctxsep \Delta, x:\uimage(\Delta') \vdash PQ : \beta }
			\end{prooftree}} \]
		where $\rb\gamma = [\gamma_0, \dots, \gamma_n]$,
		$\Gamma = \Gamma_0 + \Gamma_1$
		and $\Gamma_1 = \Gamma_{1,1} + \dots + \Gamma_{1,n}$,
		$\rb{\alpha} = \rb{\alpha}_0 + \rb{\alpha}_1$
		and $\rb{\alpha}_1 = \rb{\alpha}_{1,1} + \dots + \rb{\alpha}_{1,n}$,
		$\Delta = \Delta_0 \vee \Delta_1$,
		the sets $X_0$ and $X_1$ are equal either to $\emptyset$
		or to $\uimage(\Delta')$ and cannot be both empty,
		and the hypothesis $(\dots)$ is the coherence assumption
		$(\Delta_0, x:X_0) \coh (\Delta_1, x:X_1)$.
		We denote by $K$ the index set
		$\set{0} \cup \set{(1,i)}[1 \leq i \leq n]$,
		and for each $k \in K$ we denote by 
		\begin{equation} \label{lem:typing:fw-subst-lemma-nonempty:eq:hypM}
			\Gamma_k, x:\rb{\alpha}_k \ctxsep \Delta_k, x:X_k
			\vdash R_k:\epsilon_k
		\end{equation}
		the corresponding hypotheses of the above derivation.
		In addition, the derivation
		$\Gamma' \ctxsep \Delta' \vdashbang N : \rb\alpha$
		must be obtained by the rule \drule{!^+}, with hypotheses
		\begin{equation} \label{lem:typing:fw-subst-lemma-nonempty:eq:hypN}
			\Gamma'_j \ctxsep \Delta' \vdash N:\alpha_j
		\end{equation}
		for each element $\alpha_j$ of $\rb{\alpha}$.
		Then, for each $k \in K$:
		\begin{itemize}
		\item If $X_k = \varnothing$ (and thus $\rb{\alpha}_k = []$)
			then by \cref{lem:typing:domain-subset-domain-eq-fv}
			we have $x \notin \fv(R_k)$,
			therefore the hypothesis 
			in \cref{lem:typing:fw-subst-lemma-nonempty:eq:hypM} becomes
			$\Gamma_k + \Gamma'_k \ctxsep \Delta_k
			\vdash R_k \subst N = R_k : \epsilon_k$,
			where $\Gamma'_k$ is the empty environment.
		\item If $X_k = \uimage(\Delta')$ and $\rb{\alpha}_k = []$
			then by \cref{lem:typing:domain-subset-domain-eq-fv}
			and rule \drule{!^0} there is a derivation
			$\Gamma'_k \ctxsep \Delta' \vdashbang N:[]_{\uimage(\Delta')}$,
			where $\Gamma'_k$ is the empty environment.
			By \cref{lem:typing:fw-subst-lemma-empty}
			applied to this derivation, the derivation in 
			\cref{lem:typing:fw-subst-lemma-nonempty:eq:hypM},
			and the coherence assumption $\Delta_j \coh \Delta'$,
			there is a derivation
			$\Gamma_k + \Gamma'_k \ctxsep \Delta_k \vee \Delta'
			\vdash R_k \subst N : \epsilon_k$.
		\item If $X_k = \uimage(\Delta')$ and $\rb{\alpha}_k \neq []$
			then by applying the rule \drule{!^+}
			to all hypotheses $\Gamma'_j \ctxsep \Delta' \vdash N:\alpha_j$
			from \cref{lem:typing:fw-subst-lemma-nonempty:eq:hypN}
			such that $\alpha_j$ is sent into $\rb{\alpha}_k$
			in the decomposition $\rb{\alpha} = \rb{\alpha}_0 
			+ \rb{\alpha}_{1,0} + \dots + \rb{\alpha}_{1,n}$,
			we obtain a derivation
			$\Gamma'_k \ctxsep \Delta' \vdashbang N:\rb{\alpha}_k$.
			By the induction hypothesis
			applied to this derivation, the derivation in
			\cref{lem:typing:fw-subst-lemma-nonempty:eq:hypM}
			and the coherence assumption $\Delta_j \coh \Delta'$,
			there is a derivation
			$\Gamma_k + \Gamma'_k \ctxsep \Delta_k \vee \Delta'
			\vdash R_k \subst N : \epsilon_k$.
		\end{itemize}
	To conclude, we can apply the rules \drule{!^+} and \drule{@} again,
	as in the derivation above,
	in each of the three possible cases
	($X_0 = \varnothing$ and $X_1 = \uimage(\Delta')$,
	$X_0 = \uimage(\Delta')$ and $X_1 = \varnothing$, and
	$X_0 = X_1 = \uimage(\Delta')$).
	In each case, by the facts that
	$\Gamma = \sum_{k \in K} \Gamma_k$
	and $\Gamma' = \sum_{k \in K} \Gamma'_k$
	and by the associativity and the idempotency of $\vee$,
	we obtain $\Gamma + \Gamma' \ctxsep \Delta \vee \Delta'
	\vdash (P \subst N)(Q \subst N) : β$.
	\qedhere
	\end{itemize}
	\end{proof}

\begin{thm}[Subject Reduction] \label{thm:typing:subject-reduction}
	If $\Gamma \ctxsep \Delta \vdash M : \sigma$
	and $M \bred N$
	then $\Gamma \ctxsep \Delta \vdash N : \sigma$.
\end{thm}
	
	\begin{proof}
	We proceed by induction on the reduction $M \bred N$.
	
	In the case of a root step $(λx.P)Q \bred P \subst Q$,
	there are three possible cases for the derivation
	$\Gamma \ctxsep \Delta \vdash (λx.P)Q : \sigma$.
	\begin{itemize}
	\item The first possibility is:
		\[	\begin{prooftree}
			\hypo{ \domain(\Delta) = \fv((λx.P)Q) }
			\infer1[\esum{}]{ \ctxsep \Delta \vdash\esum{\uimage(\Delta)}. }
			\end{prooftree} \]
		Since we are working in $\LamI$
		we have $\fv(P \subst Q) = \fv((λx.P)Q)$, hence we can just replace
		$(λx.P)Q$ with $P \subst Q$ in this derivation
		and obtain the desired result.
	\item The second possibility is:
		\[	\begin{prooftree}
			\hypo{\vdots}
			\infer1{ \Gamma \ctxsep \Delta_0, x:\uimage(\Delta_1)
				\vdash P:β }
			\infer1[λ^0]{ \Gamma \ctxsep \Delta_0
				\vdash λx.P : []_{\uimage(\Delta_1)} \lto β }
			\hypo{ \domain(\Delta_1) = \fv(Q) }
			\infer1[!^0]{ \ctxsep \Delta_1
				\vdashbang Q : []_{\uimage(\Delta_1)} }
			\hypo{ \Delta_0 \coh \Delta_1 }
			\infer3[@]{ \Gamma \ctxsep \Delta_0 \vee \Delta_1 
				\vdash (λx.P)Q : β }
			\end{prooftree} \]
		and \cref{lem:typing:fw-subst-lemma-empty}
		yields the desired derivation
		$\Gamma \ctxsep \Delta_0 \vee \Delta_1 \vdash P \subst Q : β$.
	\item Similarly, the third possibility corresponds to
		a rule \drule{@} whose first and second hypothesis
		are rules \drule{λ^+} and \drule{!^+},
		and the desired result follows immediately
		by \cref{lem:typing:fw-subst-lemma-nonempty}.
	\end{itemize}
	
	All the other cases are straightforward.
	\end{proof}

\begin{lem}[Backward Substitution Lemma]
\label{lem:typing:bw-subst-lemma}
	If there is a derivation $\Gamma \ctxsep \Delta \vdash M \subst N : β$
	with $x \in \fv(M)$,
	then there are $\rb{\alpha} \in \MTypes$ and derivations
	\[	\Gamma^M, x:\rb{\alpha} \ctxsep \Delta^M, x:\uimage(\Delta^N)
			\vdash M : β
		\qquad \text{and} \qquad
		\Gamma^N \ctxsep \Delta^N \vdashbang N : \rb{\alpha} \]
	such that $\Gamma = \Gamma^M + \Gamma^N$,
	$\Delta^M \coh \Delta^N$ and $\Delta = \Delta^M \vee \Delta^N$.
\end{lem}
	
	\begin{proof}
	We proceed by induction on $M$.
	
	\begin{itemize}
	\item If $M = x$ then we can take $\rb{\alpha} \eqdef [β]$.
		For the first derivation we take $\Gamma^M$ and $\Delta^M$ empty
		and we apply the rule \drule{ax}.
		For the second derivation we take $\Gamma^N \eqdef \Gamma$,
		$\Delta^N \eqdef \Delta$ and we apply the rule \drule{!^+}
		under the hypothesis.
	
	\item If $M = λy.P$ for $y \neq x$ and $y \notin \fv(N)$,
		then there are two possibilities.
		The first one is when $β = []_Y \lto \delta$
		and the hypothesis has been derived with the rule \drule{λ^0}
		from the hypothesis
		$\Gamma \ctxsep \Delta, y:Y \vdash P \subst N : \delta$.
		By induction there are $\rb{\alpha} \in \MTypes$ and derivations
		\[	\Gamma^P, x:\rb{\alpha} \ctxsep \Delta^P, x:\uimage(\Delta^N)
				\vdash P : \delta
			\qquad \text{and} \qquad
			\Gamma^N \ctxsep \Delta^N \vdashbang N : \rb{\alpha} \]
		such that $\Gamma = \Gamma^P + \Gamma^N$,
		$\Delta^P \coh \Delta^N$ and $\Delta, y:Y = \Delta^P \vee \Delta^N$.
		
		Define $\Gamma^M \eqdef \Gamma^P$, and observe that
		$y \notin \support(\Gamma) \supseteq \support(\Gamma^M)$.
		Since $M \in \LamI$ we have $y \in \fv(P) = \domain(\Delta^P)$
		and $\Delta^P(y) = Y$ by the above conditions,
		hence we can write $\Delta^P = \Delta^M, y:Y$
		for some $\Delta^M$.
		If we apply the rule \drule{λ^0} again
		under the first derivation above, we obtain 
		$\Gamma^M, x:\rb{\alpha} \ctxsep \Delta^M, x:\uimage(\Delta^N)
		\vdash λy.P : []_Y \lto \delta$, which is the desired result.
		The decompositions of $\Gamma$ and $\Delta$ follow easily
		from those obtained by induction.
		
		The second case is when $β = \rb{\gamma} \lto \delta$
		and the initial hypothesis has been derived 
		with the rule \drule{λ^+}. The proof is exactly identical.
	
	\item If $M = PQ$, then again there are two possibilities.
		The first one is when the hypothesis has been derived as follows:
		\[	\begin{prooftree}
			\hypo{\vdots}
			\infer1{ \Gamma \ctxsep \Delta_0
				\vdash P \subst N : []_{\uimage(\Delta_1)} \lto \beta }
			\hypo{ \domain(\Delta_1) = \fv(Q \subst N) }
			\infer1[!^0]{ \ctxsep \Delta_1
				\vdashbang Q \subst N : []_{\uimage(\Delta_1)} }
			\hypo{ \Delta_0 \coh \Delta_1 }
			\infer3[@]{ \Gamma \ctxsep \Delta
				\vdash P \subst N Q \subst N : \beta }
			\end{prooftree} \]
		where $\Delta = \Delta_0 \vee \Delta_1$.
		We perform the following constructions.
		
		\begin{itemize}
		\item If $x \notin \fv(P)$,
			then the first hypothesis can be presented as
			$\Gamma^P \ctxsep \Delta_0^P
			\vdash P : []_{\uimage(\Delta_1)} \lto \beta$,
			with $\Gamma^P \eqdef \Gamma$ and $\Delta_0^P \eqdef \Delta_0$.
			
			Otherwise $x \in \fv(P)$,
			and by induction on the first hypothesis
			we obtain some $\rb{\alpha} \in \MTypes$ and derivations
			\[	\Gamma^P, x:\rb{\alpha} 
					\ctxsep \Delta_0^P, x:\uimage(\Delta_0^N)
					\vdash P : []_{\uimage(\Delta_1)} \lto \beta
				\qquad \text{and} \qquad
				\Gamma^N \ctxsep \Delta_0^N \vdashbang N : \rb{\alpha} \]
			such that $\Gamma = \Gamma^P + \Gamma^N$,
			$\Delta_0^P \coh \Delta_0^N$ and 
			$\Delta_0 = \Delta_0^P \vee \Delta_0^N$.
			
		\item If $x \notin \fv(Q)$,
			then the second hypothesis can be presented as
			$\ctxsep \Delta_1^Q \vdashbang Q:[]_{\uimage(\Delta_1)}$,
			with $\Delta_1^Q \eqdef \Delta_1$.
			
			Otherwise $x \in \fv(Q)$ and $\domain(\Delta_1) = 
			\fv(Q \subst N) = \fv(N) \setminus \set x \cup \fv(Q)$.
			Therefore we can write a decomposition
			$\Delta_1 = \Delta_1^Q \vee \Delta_1^N$ where
			$\Delta_1^Q$ (resp. $\Delta_1^N$) is defined
			by restricting the domain of $\Delta_1$
			to $\fv(Q) \setminus \set x$ (resp. to $\fv(N)$).
			Observe that
			\[	\uimage(\Delta_1^Q, x:\uimage(\Delta_1^N))
				= \uimage(\Delta_1^Q) \cup \uimage(\Delta_1^N)
				= \uimage(\Delta_1), \]
			hence we can derive
			\[	\begin{prooftree}
				\hypo{ \domain(\Delta_1^Q, x:\uimage(\Delta_1^N) = \fv(Q) }
				\infer1[!^0]{ \ctxsep \Delta_1^Q, x:\uimage(\Delta_1^N)
					\vdashbang Q : []_{\uimage(\Delta_1)}. }
				\end{prooftree} \]
		\end{itemize}
		Observe that whenever $x \in \fv(P)$ and $x \in \fv(Q)$,
		we define both $\Delta_0^N$ and $\Delta_1^N$.
		However, since
		$\Delta_0^N \subseteq \Delta_0 \coh \Delta_1 \supseteq \Delta_1^N$
		and $\domain(\Delta_0^N) = \domain(\Delta_1^N) = \fv(N)$,
		they are in fact equal.
		Therefore, in all cases we denote  $\Delta_0^N$ and $\Delta_1^N$
		simply by $\Delta^N$.
		
		To conclude, there are three possibles cases
		($x \notin \fv(P)$ and $x \notin \fv(Q)$ cannot occur together
		since $x \in \fv(PQ)$):
		\begin{itemize}
		\item If $x \notin \fv(P)$ and $x \in \fv(Q)$,
			define $\rb{\alpha} \eqdef []_{\uimage(\Delta^N)}$ so that
			$\ctxsep \Delta^N \vdashbang N:\rb{\alpha}$.
			Observe that
			$\Delta_0^P = \Delta_0 \coh \Delta_1 \supseteq \Delta_1^Q$
			and $x \notin \domain(\Delta_1^Q)$,
			hence we can apply the rule \drule{@} and obtain a derivation
			\begin{equation} \label{lem:typing:bw-subst-lemma:eq:app1}
				\Gamma^P, x:\rb{\alpha}
				\ctxsep \Delta_0^P \vee \Delta_1^Q, x:\uimage(\Delta^N)
				\vdash PQ : β.
			\end{equation}
		\item If $x \in \fv(P)$ and $x \notin \fv(Q)$, observe that 
			$\Delta_0^P \subseteq \Delta_0 \coh \Delta_1 = \Delta_1^Q$
			and $x \notin \fv(Q) = \domain(\Delta_1) = \domain(\Delta_1^Q)$,
			hence we can apply the rule \drule{@} and obtain exactly
			the derivation from \cref{lem:typing:bw-subst-lemma:eq:app1}.
		\item If $x \in \fv(P)$ and $x \in \fv(Q)$, we observe that
			$\Delta_0^P \subseteq \Delta_0
			\coh \Delta_1 \supseteq \Delta_1^Q$
			and $x \notin \fv(Q) = \domain(\Delta_1^Q)$,
			and we apply the rule \drule{@}, obtaining
			the derivation from \cref{lem:typing:bw-subst-lemma:eq:app1}.
		\end{itemize}
		In all three cases we also have
		$\Gamma = \Gamma^P + \Gamma^N$
		($\Gamma_N$ being the empty environment in the cases where
		we did not explicitly define it),
		$(\Delta_0^P \vee \Delta_1^Q) \coh \Delta^N$ and 
		$\Delta = (\Delta_0^P \vee \Delta_1^Q) \vee \Delta^N$,
		hence we can conclude with $\Gamma^M \eqdef \Gamma^P$
		and $\Delta^M \eqdef \Delta_0^P \vee \Delta_1^Q$.
		
		The second possibility is when the initial hypothesis
		has been derived as follows:
		\[	\begin{prooftree}
			\hypo{\vdots}
			\infer1{ \Gamma_0 \ctxsep \Delta_0
				\vdash P \subst N : \rb{\gamma} \lto \beta }
			\hypo{\vdots}
			\infer1{ \Gamma_{1,i} \ctxsep \Delta_1 
				\vdash Q \subst N : \gamma_i }
			\delims{ \left[ }{ \right]_{i=0}^n }
			\infer[left label={!^+}]1{ \Gamma_1 \ctxsep \Delta_1
				\vdashbang Q \subst N : \rb{\gamma} }
			\hypo{ \Delta_0 \coh \Delta_1 }
			\infer3[@]{ \Gamma \ctxsep \Delta
				\vdash P \subst N Q \subst N : \beta }
			\end{prooftree} \]
		where $\rb\gamma = [\gamma_0, \dots, \gamma_n]$,
		$\Gamma = \Gamma_0 + \Gamma_1$
		and $\Gamma_1 = \Gamma_{1,1} + \dots + \Gamma_{1,n}$,
		and $\Delta = \Delta_0 \vee \Delta_1$.
		We denote by $K$ the index set
		$\set{0} \cup \set{(1,i)}[1 \leq i \leq n]$,
		and for each $k \in K$ we denote by 
		\begin{equation} \label{lem:typing:bw-subst-lemma:eq:app2-hypM}
			\Gamma_k \ctxsep \Delta_k \vdash R_k \subst N:\epsilon_k
		\end{equation}
		the corresponding hypotheses of the above derivation.
		Then for each $k \in K$ we perform the following constructions.
		\begin{itemize}
		\item If $x \notin \fv(R_k)$, then
			\cref{lem:typing:bw-subst-lemma:eq:app2-hypM}
			can be presented as
			$\Gamma_k^{R_k} \ctxsep \Delta_ k^{R_k}
			\vdash R_k : \epsilon_k$,
			with $\Gamma_k^{R_k} \eqdef \Gamma_k$
			and $\Delta_k^{R_k} \eqdef \Delta_k$.
			We also define $\Gamma_k^N$ to be the empty environment
			and $\rb{\alpha}_k \eqdef []_{\uimage(\Delta_k^N)}$.
		\item Otherwise $x \in \fv(R_k)$, and by induction on
			\cref{lem:typing:bw-subst-lemma:eq:app2-hypM}
			we obtain some $\rb{\alpha}_k \in \MTypes$ and derivations
			\[	\Gamma_k^{R_k}, x:\rb{\alpha}_k
					\ctxsep \Delta_k^{R_k}, x:\uimage(\Delta_k^N)
					\vdash R_k : \epsilon_k
				\qquad \text{and} \qquad
				\Gamma_k^N \ctxsep \Delta_k^N \vdashbang N : \rb{\alpha}_k \]
			such that $\Gamma_k = \Gamma_k^{R_k} + \Gamma_k^N$,
			$\Delta_k^{R_k} \coh \Delta_k^N$ and 
			$\Delta_k = \Delta_k^{R_k} \vee \Delta_k^N$.
		\end{itemize}
		In addition, let us make the following observations.
		\begin{itemize}
		\item If $x \in \fv(Q)$ then for all $i,j \in [1,n]$ we have
			$\Delta_{1,i}^Q \subseteq \Delta_1 \supseteq \Delta_{1,j}^Q$
			and $\domain(\Delta_{1,i}^Q) = \domain(\Delta_{1,j}^Q)
			= \fv(Q) \setminus \set x$,
			hence $\Delta_{1,i}^Q = \Delta_{1,j}^Q$
			and we can denote by $\Delta_1^Q$ the unique value taken
			by all $\Delta_{1,i}^Q$.
			If $x \notin \fv(Q)$ we define $\Delta_ 1^Q$ to be empty.
		\item Whenever $x \in \fv(R_k)$ and $x \in \fv(R_l)$,
			$\Delta_k^N = \Delta_l^N$, because
			$\Delta_k^N \subseteq \Delta_k 
			\coh \Delta_l \supseteq \Delta_l^N$
			(the coherence is due to the fact that $\Delta_k$ and $\Delta_l$
			are either $\Delta_0$ or $\Delta_1$)
			and $\domain(\Delta_k^N) = \domain(\Delta_l^N) = \fv(N)$.
			Therefore we denote by $\Delta^N$ the unique value taken
			by all defined $\Delta_k^N$.
			Notice that we defined at least one of these,
			since $x \in \fv(M)$ implies that 
			$x \in \fv(P)$ or $x \in \fv(Q)$.
		\end{itemize}
		As a consequence we can apply again the rules
		\drule{!^+} and \drule{@} and obtain:
		\[	\sum_{k \in K} \Gamma_k^{R_k}, x:\sum_{k \in K} \rb{α}_k 
			\ctxsep \Delta_0^P \vee \Delta_1^Q, x:\uimage(\Delta^N)
			\vdash PQ : β \]
		and, by \enquote{unwrapping} all derivations
		$\Gamma_k^N \ctxsep \Delta_k^N \vdashbang N : \rb{\alpha}_k$
		using the rules \drule{!^0} and \drule{!^+},
		then \enquote{rewrapping} them using the rule \drule{!^+}:
		\[	\sum_{k \in K} \Gamma_k^N \ctxsep \Delta^N
			\vdashbang N : \sum_{k \in K} \rb{\alpha}_k \]
		which are the desired derivations, with
		$\Gamma^M \eqdef \sum_{k \in K} \Gamma_k^{R_k}$,
		$\rb{α} \eqdef \sum_{k \in K} \rb{α}_k$,
		$\Delta^M \eqdef \Delta_0^P \vee \Delta_1^Q$, and
		$\Gamma^N \eqdef \sum_{k \in K} \Gamma_k^N$.
		The conditions $\Gamma = \Gamma^M + \Gamma^N$,
		$\Delta^M \coh \Delta^N$ and $\Delta = \Delta^M \vee \Delta^N$.
		follow immediately from the above constructions.
	\qedhere
	\end{itemize}
	\end{proof}

\begin{thm}[Subject Expansion] \label{thm:typing:subject-expansion}
	If $\Gamma \ctxsep \Delta \vdash N : \sigma$
	and $M \bred N$
	then $\Gamma \ctxsep \Delta \vdash M : \sigma$.
\end{thm}

	\begin{proof}
	We proceed by induction on the reduction $M \bred N$,
	exactly as for subject reduction \cref{thm:typing:subject-reduction}.
	
	In the case of a root step $(λx.P)Q \bred P \subst Q$,
	there are two possible cases for the derivation
	$\Gamma \ctxsep \Delta \vdash P \subst Q : \sigma$.
	\begin{itemize}
	\item The first possibility is:
		\[	\begin{prooftree}
			\hypo{ \domain(\Delta) = \fv(P \subst Q) }
			\infer1[\esum{}]{ \ctxsep \Delta \vdash\esum{\uimage(\Delta)}. }
			\end{prooftree} 
		\]
		Since we are working in $\LamI$
		we have $\fv(P \subst Q) = \fv((λx.P)Q)$, hence we can just replace
		$P \subst Q$ with $(λx.P)Q$ in this derivation
		and obtain the desired result.
	\item The second possibility is when $\sigma$ is some $β \in \Types$.
		Since we are working in $\LamI$ we have $x \in \fv(P)$,
		hence we can apply \cref{lem:typing:bw-subst-lemma}
		to the derivation $\Gamma \ctxsep \Delta \vdash P \subst Q : β$
		and obtain $\rb{\alpha} \in \MTypes$ and derivations
		$\Gamma^P, x:\rb{\alpha} \ctxsep \Delta^P, x:\uimage(\Delta^Q)
		\vdash P : β$ and
		$\Gamma^Q \ctxsep \Delta^Q \vdashbang Q : \rb{\alpha}$
		such that $\Gamma = \Gamma^P + \Gamma^Q$,
		$\Delta^P \coh \Delta^Q$ and $\Delta = \Delta^P \vee \Delta^Q$.
		Then we build the following derivation:	
		\[	\begin{prooftree}
			\hypo{\vdots}
			\infer1{ \Gamma^P, x:\rb{\alpha}
				\ctxsep \Delta^P, x:\uimage(\Delta^Q) \vdash P:β }
			\infer[right label template	= {\small\inserttext}]1%
				[\drule{λ^0} or \drule{λ^+}]
				{ \Gamma^P \ctxsep \Delta^P
				\vdash λx.P : \rb{\alpha} \lto β }
			\hypo{ \vdots }
			\infer1{ \Gamma^Q \ctxsep \Delta^Q \vdashbang Q : \rb{\alpha} }
			\hypo{ \Delta^P \coh \Delta^Q }
			\infer3[@]{ \Gamma \ctxsep \Delta \vdash (λx.P)Q : β. }
			\end{prooftree} \]
	\end{itemize}
	
	All the other cases are straightforward.
	\end{proof}

\subsection{The multi-type interpretation characterises $\cO$}

We now exploit our typing system to provide a characterization of $\cO$,
introducing the following interpretation
and the subsequent theorem.

\begin{defi}
	For all $M \in \LamI$, its \emph{multi-type interpretation}
	is the set
	\[	\interp{M} \eqdef \set{ (\Gamma,\Delta,\sigma) }
			[ \Gamma \ctxsep \Delta \vdash M : \sigma ]. \]
\end{defi}

\begin{thm} \label{thm:interp-characterises-O}
	For all $M,N \in \LamI$, 
	$\interp M = \interp N$ if and only if $\OT M = \OT N$.
\end{thm}

We divide the proof of the theorem in two parts,
\cref{prop:interp-characterises-O-1,prop:interp-characterises-O-2}.
For the former we rely again on Taylor expansion,
following the well-established connexion between
the Taylor expansion of a λ-term and its typing derivations
in the usual multi-type system \autocite[§~6.3]{deCarvalho.07},
which was already exploited by \autocite[Corollary~3.11]{Manzonetto.Ruo.14}
to obtain the same result for this system and the \lam-theory $\cB$.

\begin{prop} \label{prop:interp-characterises-O-1}
	For all $M,N \in \LamI$, 
	if $\OT M \MTle \OT N$ then $\interp M \subseteq \interp N$.
\end{prop}

To use the Taylor expansion as a way to connect Ohana trees
and typing derivations,
let us introduce the following slightly modified version
of the typing rules from \cref{def:typing-rules},
adapted to \lamI-resource terms:
\begin{gather*}
	\begin{prooftree}
	\infer0[ax_{\resource}]{ x : [\alpha] \ctxsep x : X \vdash x : \alpha }
	\end{prooftree}
\qquad
	\begin{prooftree}
	\hypo{ \Gamma_0 \ctxsep \Delta_0 \vdash s : \rb\alpha \lto \beta }
	\hypo{ \Gamma_1 \ctxsep \Delta_1 \vdashbang \rb t : \rb\alpha }
	\hypo{ \Delta_0 \coh \Delta_1 }
	\infer3[@_{\resource}]
		{ \Gamma_0 + \Gamma_1 \ctxsep \Delta_0 \vee \Delta_1
		\vdash s \rb t : \beta }
	\end{prooftree}
\\[\topsep]
	\begin{prooftree}
	\hypo{ \Gamma, x : [] \ctxsep \Delta, x : X \vdash s : \beta }
	\infer1[λ^0_{\resource}]
		{ \Gamma \ctxsep \Delta \vdash λx.s : []_X \lto \beta }
	\end{prooftree}
\qquad
	\begin{prooftree}
	\hypo{ \Gamma, x : [\alpha_0, \dots, \alpha_n] \ctxsep \Delta, x : X 
		\vdash s : \beta }
	\infer1[λ^+_{\resource}]{ \Gamma \ctxsep \Delta 
		\vdash λx.s : [\alpha_0, \dots, \alpha_n] \lto \beta }
	\end{prooftree}
\\[\topsep]
	\begin{prooftree}
	\infer0[!^0_{\resource}]
		{ \ctxsep \Delta \vdashbang
		[]_{\domain(\Delta)} : []_{\uimage(\Delta)} }
	\end{prooftree}
\qquad
	\begin{prooftree}
	\hypo{ \Gamma_0 \ctxsep \Delta \vdash t_0 : \alpha_0 }
	\hypo{ \dots }
	\hypo{ \Gamma_n \ctxsep \Delta \vdash t_n : \alpha_n }
	\infer3[!^+_{\resource}]
		{ \Gamma_0 + \dots + \Gamma_n \ctxsep \Delta \vdashbang 
		[t_0, \dots, t_n] : [\alpha_0, \dots, \alpha_n] }
	\end{prooftree}
\end{gather*}
Just as for \lamI-terms, we write
$\pi \derives \Gamma \ctxsep \Delta \vdash s : \alpha$
whenever $\pi$ is a derivation of the given conclusion,
and simply $\Gamma \ctxsep \Delta \vdash s : \alpha$
to express that such a derivation exists.

\begin{lem} \label{lem:typing:resource-subject-reduction-expansion}
	The following subject reduction and subject expansion properties hold:
	\begin{itemize}
	\item If $\Gamma \ctxsep \Delta \vdash \rts : \alpha$
		and $\rts \rreds \rs t$,
		then we can write $\rs t = \rtt + \rs t'$
		with $\Gamma \ctxsep \Delta \vdash \rtt : \alpha$.
	\item If $\Gamma \ctxsep \Delta \vdash \rtt : \alpha$
		and $\rts \rreds \rtt + \rs t'$
		then $\Gamma \ctxsep \Delta \vdash \rts : \alpha$.
	\end{itemize}
\end{lem}
	
	\begin{proof}
	The one-step version of these properties
	(\ie where $\rreds$ is replaced with $\rred$)
	can be proved following exactly the same path as
	for \cref{thm:typing:subject-reduction,thm:typing:subject-expansion},
	respectively.
	It can then be extended to the reflexive-transitive closure $\rreds$
	by a straightforward induction.
	\end{proof}

\begin{lem} \label{lem:typing:approx-TE-typings}
	$\Gamma \ctxsep \Delta \vdash M : \alpha$
	if and only if there exists $\rts \TEin \TE M$
	such that $\Gamma \ctxsep \Delta \vdash \rts : \alpha$.
\end{lem}
	
	\begin{proof}
	Both directions are proved by induction on the given derivations.
	The construction is completely transparent:
	rule \drule{ax} corresponds to rule \drule{ax_{\resource}},
	rule \drule{@} corresponds to rule \drule{@_{\resource}},
	and so on.
	In the application cases, the size of the multisets
	in both derivations and in the resource approximant $\rts \TEin \TE M$
	are of course identical.
	\end{proof}

Notice that this Lemma bears the content of an
\enquote{Approximation Theorem}
in the sense of \cite[Theorem~3.10]{Manzonetto.Ruo.14}:
if, for all $(X,\mathcal X)  \in \bigdunion_{X \subf \Var} 
\parts{\rterms(X)}$, one defines
\[	\interp{X, \mathcal X} \eqdef
	\set*{ (\Gamma, \Delta, \esum{\uimage(\Delta)}) }[\begin{array}{@{}l@{}}
		\support(\Gamma)=  \emptyset \\ \domain(\Delta) = X \end{array}]
	\cup \set{(\Gamma,\Delta,\alpha)}[\exists \rts \in \mathcal X,\ 
		\Gamma \ctxsep \Delta \vdash \rts : \alpha] \]
then \cref{lem:typing:approx-TE-typings} implies that
for all $M \in \LamI$, $\interp M = \interp{\TE M}$.
This approximation is the key ingredient of the following proof.

\begin{proof}[Proof of \cref{prop:interp-characterises-O-1}]
	Take two terms $M,N \in \LamI$, then:
	\begin{flalign*}
	& \OT M \MTle \OT N \\
	\Rightarrow\quad & \TE{\OT M} \TElt \TE{\OT N}
		& \llap{by \cref{lem:TE-monotone},} \\
	\Rightarrow\quad & \nf(\TE M) \TElt \nf(\TE N)
		& \llap{by \cref{the:commutation},} \\
	\Rightarrow\quad & \left\{ \begin{array}{l}
		\fv(M) = \fv(N) \\
		\interp{M}' \subseteq \interp{N}' \text{ for }
			\interp{M}' \eqdef \set{ (\Gamma,\Delta,\alpha) }
				[ \exists \rtt \TEin \nf(\TE M),\ 
				\Gamma \ctxsep \Delta \vdash \rtt : \alpha ]
		\end{array} \right. \\
	\Rightarrow\quad & \left\{ \begin{array}{l}
		\fv(M) = \fv(N) \\
		\interp{M}'' \subseteq \interp{N}'' \text{ for }
			\interp{M}'' \eqdef \set{ (\Gamma,\Delta,\alpha) }
				[ \exists \rts \TEin \TE M,\ 
				\Gamma \ctxsep \Delta \vdash \rts : \alpha ]
		\end{array} \right. \\
		&& \llap{by \cref{lem:typing:resource-subject-reduction-expansion},}
		\\
	\Rightarrow\quad & \left\{ \begin{array}{l}
		\interp{M}'''_0 = \interp{N}'''_0 \text{ for }
			\interp{M}'''_0 \eqdef \set{ (\Gamma,\Delta,\esum X) }
				[ \Gamma \ctxsep \Delta \vdash M : \esum X ] \\
		\interp{M}'''_1 \subseteq \interp{N}'''_1 \text{ for }
			\interp{M}'''_1 \eqdef \set{ (\Gamma,\Delta,\alpha) }
				[ \Gamma \ctxsep \Delta \vdash M : \alpha ]
		\end{array} \right. \\
		&& \llap{by \cref{lem:typing:domain-subset-domain-eq-fv,%
			lem:typing:approx-TE-typings},}
		\\
	\Rightarrow\quad & \interp M \subseteq \interp N. & \qedhere
	\end{flalign*}
\end{proof}

The second part of \cref{thm:interp-characterises-O} relies on a different argument, which is already implicit in \autocite[Theorem~3]{Rocca82} and was subsequently employed by \autocites[Lemma~5.5]{BreuvartMR18}[Lemma~14.75]{BarendregtM22}, whose approach we follow, as well as by \autocite[§~6]{Lancelot.25} in a coinductive setting.

\begin{prop} \label{prop:interp-characterises-O-2}
	For all $M,N \in \LamI$, 
	if $\OT M \neq \OT N$ then $\interp M \neq \interp N$.
\end{prop}

	\begin{proof}
	\NewDocumentCommand \Deltaid {O{M}} { \Delta_{\mathrm{id}}^{#1} }
	Take $M,N \in \LamI$ such that $\OT M \neq \OT N$,
	we want to build a derivation
	$\pi \derives \Gamma \ctxsep \Deltaid \vdash M : \sigma$
	such that there exists no derivation
	$\pi' \derives \Gamma \ctxsep \Deltaid \vdash N : \sigma$
	(without loss of generality, since $M$ and $N$ play symmetric roles),
	where $\Deltaid$ is a notation for 
	the partial map $x \mapsto \set x$ defined only on $\fv(M)$.
	
	We do so by induction on the depth of the first
	difference between $\OT M$ and $\OT N$
	(since Ohana trees are defined by coinduction,
	their equality is a coinductive construction
	and the negation of it is inductive).
	
	\begin{itemize}
	\item If the first difference occurs at depth~0, there are three cases.
		\begin{itemize}
		\item If $M \hreds λx_1\dots x_n.yM_1 \cdots M_k$
			then we can derive either
			\[	y:[[]_{\fv(M_1)} \lto \cdots \lto []_{\fv(M_k)} \lto \atom]
				\ctxsep \Deltaid \vdash M : 
					\underbrace{[]_\emptyset \lto \cdots \lto []_\emptyset} 
					_{\text{$n$ times}} \lto \atom, \]
			if $y$ is not among the $x_1,\dots, x_n$, or
			\[	\ctxsep \Deltaid \vdash M : 
				\underbrace{[]_\emptyset \lto \cdots \lto 
					[[]_{\fv(M_1)} \lto \cdots \lto []_{\fv(M_k)} \lto \atom]
					\lto \cdots \lto []_\emptyset} 
				_{\text{$n$ times}} \lto \atom, \]
			if $y = x_i$ for some $i\in\set{1,\dots,n}$, in which case the non-empty multiset
			occurs in the $i$-th argument position in the given type.
			In the two possible situations:
			\begin{itemize}
			\item $N$ has no head normal form. 
			Since the type of $M$ is some $\alpha \in \Types$, we can apply \cref{lem:typing:approx-TE-typings} and obtain that 
			$\Gamma\ctxsep\Delta\vdash N : \alpha$ is derivable iff there is $\rts \TEin \TE N$ such that 
			$\Gamma \ctxsep \Delta \vdash \rts : \alpha$. By strong normalization (\cref{the:rred-norm-confl}) and subject reduction (\cref{lem:typing:resource-subject-reduction-expansion}), 
			we would get an $\rts' \in^\bullet\NFT{N}$ such that $\Gamma \ctxsep \Delta \vdash \rts' : \alpha$. 
			Impossible, since $\NFT{N} = \TMT{N} = \TE{\bot_{\fv(N)}} = (\fv(N),\emptyset)$, by the Commutation \cref{the:commutation}.
			\item $N$ has a head normal form of a different shape. The only interesting case is when
			\[
				N \hreds λx_1\dots x_{n'}.yN_1 \cdots N_{k'}
			\]
			with $n-k = n'-k'\neq 0$ (this happens, \eg, when $M$ is an $\eta$-expansion of $N$). 
			By symmetry we can assume that $n > n'$ and $k>k'$, so it is clear that $N$ cannot be given a derivation with same environment and type.
			\end{itemize}
		\item In case $N$ has a head normal form while $M$ does not, the same argument applies.
		\item Finally, if both $M$ and $N$ do not have a head normal form 
			and $\fv(M) \neq \fv(N)$,
			then we can derive
			$\ctxsep \Deltaid \vdash M : \esum{\fv(M)}$,
			which cannot be done for $N$.
		\end{itemize}
	
	\item If the first difference occurs deeper, there are reductions
		$M \hreds λx_1\dots x_n.yM_1 \cdots M_k$ and
		$N \hreds λx_1\dots x_n.yN_1 \cdots N_k$,
		and for all $1 \le i \le k$, $\fv(M_i) = \fv(N_i)$.
		Let us concentrate on the case where $n = 0$
		to lighten the notations;
		the case where $n > 0$ can be deduced straightforwardly
		by analyzing the rules $\drule{λ^0}$ and $\drule{λ^+}$.
		
		Fix an $i$ such that $\OT{M_i} \neq \OT{N_i}$,
		then by induction we obtain a derivation
		$\pi_i \derives \Gamma \ctxsep \Deltaid[M_i] \vdash M_i : \alpha$
		such that there exists no derivation
		$\pi'_i \derives \Gamma \ctxsep \Deltaid[M_i] \vdash N_i : \alpha$
		(notice that the equality of the free variables of $M$ and $N$
		allows to restrict ourselves to derivations
		whose type is in $\Types$).
		For each $j \neq i$, we also build the following derivation:
		\[ \pi_j \quad\eqdef\quad \begin{prooftree}[center]
			\hypo{ \domain(\Deltaid[M_j])  = \fv(M_j) }
			\infer1[!^0]{ \ctxsep \Deltaid[M_j]
				\vdashbang M_j : []_{\fv(M_j)} }
		\end{prooftree}. \]
		
		Now let $\atom$ be a fresh atom,
		\ie it does not appear in $\pi$.
		By applying rule \drule{ax},
		then rule \drule{@} with second hypothesis $\pi_1$,
		then rule \drule{@} with second hypothesis $\pi_2$,
		and so on until we apply
		rule \drule{@} with second hypothesis $\pi_k$,
		we obtain a derivation:
		\[	\pi \derives \Gamma + ( y : [
				\underbrace{[]_{\fv(M_1)} \lto \cdots \lto [\alpha]
					\lto \cdots \lto []_{\fv(M_k)}} 
				_{\text{$n$ times}} \lto \atom
			]) \ctxsep \bigvee_{j=1}^k \Deltaid[M_j]
			\vdash yM_1\cdots M_k : \atom \]
		where the multiset $[\alpha]$ appears in $i$th argument position
		of the given type.
		
		The key observation to be made at this point is	that
		a derivation of $\vdash yN_1\cdots N_k : \atom$
		with the same environment
		\emph{must} come from the same construction than $\pi$,
		with hypotheses $\Gamma \ctxsep \Deltaid[M_i] \vdash N_i : \alpha$,
		and $\ctxsep \Deltaid[M_j] \vdashbang N_j : []_{\fv(M_j)}$
		for $j \neq i$.
		Indeed since $\atom$ has been taken fresh it \emph{only} appears
		in the type $[]_{\fv(M_1)} \lto \cdots \lto [\alpha]
		\lto \cdots \lto []_{\fv(M_k)} \lto \atom$ assigned to $y$
		in the environment (in particular it does not appear in $\Gamma$),
		hence this type of $y$ must have been used 
		to produce the final~$\atom$,
		and therefore it is the type assigned to the head occurrence of~$y$.
		This observation allows us to conclude,
		as the induction hypothesis tells us that there is no
		derivation $\Gamma \ctxsep \Deltaid[M_i] \vdash N_i : \alpha$.
	\qedhere
	\end{itemize}
	\end{proof}


\section{Conclusions and future work}

In this paper we introduced the Ohana trees for the \lamI-calculus, together with two theories of program approximations: the former based on finite trees, the latter on resource approximants and Taylor expansion. In addition, we defined a denotational model characterizing the \lamI-theory induced by Ohana trees.
Our pioneering results look encouraging, so we believe that this approach deserves further investigations.

\subsection{Further work on the \lamI-calculus}

The Ohana trees introduced in \cref{def:memorytrees4LamI} are an 
adaptation of Böhm trees, since they regard as meaningless the terms 
without hnf. 
By considering as meaningless the subset of zero-terms\footnote{\Ie, terms 
without an hnf that never reduce to an abstraction.} one obtains Lévy-Longo 
trees~\cite{Levy76,Longo83}, and by taking mute terms\footnote{Also called 
\emph{root active}, mute terms have the property that all their reducts can 
reduce to a redex.} as meaningless one obtains Berarducci 
trees~\cite{Berarducci96}. 
Our definition of Ohana trees can be readily adapted to both settings by 
appropriately modifying the notion of meaningless. 

\begin{prob}
Is it possible to obtain different \lamI-theories by varying the notion of meaningless terms underlying Ohana trees?
\end{prob}

The question is whether the equivalence induced on \lamI-terms remains contextual, and therefore a \lamI-theory.
Preliminary investigations indicate that all our results extend seamlessly to the Lévy-Longo version. 
Regarding the Berarducci version, the direct approximants can be generalized without any issues 
(except for the fact that an oracle is needed, see \cite{BakelBDV02}). 
We are currently studying whether the corresponding \lamI-resource calculus and Taylor expansion could also be designed.

Concerning denotational models, the one we constructed in \cref{sec:types} is inspired by the relational semantics of linear logic, although we have not explored its underlying categorical framework in detail. What we can say with reasonable confidence is that it does not fit any of the notions of \lamI-calculus models previously described in the literature~\cite{EgidiHR92,Jacobs93}. We believe that the following problem deserves further investigations.

\begin{prob}\label{problem:model}
Is there a categorical notion of a denotational model of $\LamI$ that is general enough to encompass all instances introduced individually in the literature?
\end{prob}

As pointed out by anonymous reviewers, our \cref{def:LamI} of \lamI-terms can be recast in the framework of abstract syntax with binding \cite{Fiore.Plo.Tur.99}---more precisely, it corresponds to the \enquote{relevant} case discussed in \cite{FioreS25}. 
Just as $\Lam(-)$ can be described as a canonical presheaf over (a skeleton of) the category of finite sets and functions, $\LamI(-)$ can be viewed as a presheaf over the category of finite sets and \emph{surjections} (intuitively, considering only surjective renamings of variables prevents any weakening). Another possibility would be to consider $\LamI(-)$ as a presheaf over the category ${\bf Bij}$ of finite sets and \emph{bijections}, so that it becomes an analytic functor.
Similarly, $\rterms**(-)$ would appear as a presheaf over the same category. This raises the question of whether all our work can be presented \enquote{over} finite sets and bijections; in particular, would our Taylor expansion act naturally with respect to these presheaf constructions?
There is hope that it is the case, given that we followed the inductive 
structures of both \lamI-terms and \lamI-resource terms, which may pave the 
way towards building (presheaf) models answering \cref{problem:model}.
Note however that choosing bijections rather than surjections induces a stratification of the corresponding notion of $\lamI$-theory that becomes a family of relations
\[
	\set{\mathcal{R}_X \subseteq \LamI(X)\times\LamI(X)}_{X\subseteq\Var} \in \textrm{\bf Bij}.
\]
While this notion of $\lamI$-theory corresponds well to the notion of \lam($\comb{K}$)-theory as defined by Hyland in \cite{Hyland17}, it does not encompass the \lamI-theories previously introduced in the literature since they equate all non-normalizable \lamI-terms, regardless of their free variables.

We believe that finding a denotational model and studying its categorical properties are important steps towards a deeper mathematical understanding of our Taylor expansion. 
In particular, it is currently unclear whether our resource calculus stands on a solid notion of differentiation, as it is the case for the usual resource calculus~\cite{EhrhardR03}, or if it is an \emph{ad hoc} adaptation.

\begin{prob}
Is the \lamI-resource calculus representing some notion of derivative?
\end{prob}

We now discuss more speculative extensions, going beyond the setting of \lamI-calculus.

\subsection{What about the full \lam-calculus?}

Our investigation of Ohana trees was originally inspired by the relational model with infinite multiplicities $\model{E}$ defined in \cite{CarraroES10}.
Indeed, this model distinguishes $\Om$ from $\Om x$, and $\comb{Y}$ from $\Bible$, just like our Ohana trees. 
Unlike our Ohana trees, it separates $\Om$ and $\Om\comb{I}$ because of the linearity of $\comb{I}$, thus the induced theory is different from Ohana trees equality.
We believe that the interpretation of \lam-terms in this model is sufficiently stratified to be translated into a coinductive notion of evaluation tree.

\begin{prob}\label{problem:model2}
Is there a notion of evaluation trees for \lam-terms capturing the equality induced by the relational model $\model{E}$ with infinite coefficients?
\end{prob}

The fact that $\model{E}$ is a model of the full \lam-calculus demonstrates the existence of consistent \lam-theories that track variables pushed into infinity in the Böhm tree semantics. 
A deeper analysis of this model may suggest ways to refine the definition of Ohana trees for the full \lam-calculus, so as to avoid the counterexample to contextuality discussed in \cref{sec:Ohana-lamK}. 
Another natural question is whether Ohana trees can serve as a meaningful notion of observation in the sense of Morris's observational equivalences~\cite{MorrisTh}.

\begin{prob}
Is there a denotational model inducing the following \lam-theory?
\[
 M \equiv N \iff	\forall C[]\,.\, \MT{C[M]} = \MT{C[N]}
\]
where $C[]$ denotes a \lam-calculus \emph{context} (namely, a \lam-term containing a hole $[]$), and $C[M]$ the \lam-term obtained by replacing $M$ for the hole $[]$ in $C[]$, possibly with capture of free variables.
\end{prob}

\section*{Acknowledgments} We would like to thank Thomas Colcombet, Thomas Ehrhard, Paul-André Melliès, and Guy McCusker for stimulating discussions that inspired the development of Ohana trees, as well as Adrienne Lancelot, Simona Ronchi della Rocca and the anonymous reviewers of this article for suggesting some technical improvements.

\printbibliography

\end{document}